\documentclass[11pt]{article}
\usepackage{graphicx}
\usepackage{pgfplots}

\usepackage[colorlinks=true,citecolor=blue]{hyperref}
\usepackage{natbib}

\usepackage{graphicx}
\makeatletter
\newcommand*\bigcdot{\mathpalette\bigcdot@{1.5}}
\newcommand*\bigcdot@[2]{\mathbin{\vcenter{\hbox{\scalebox{#2}{$\m@th#1\bullet$}}}}}

\newcommand*\smallcdot{\mathpalette\bigcdot@{.5}}
\newcommand*\smallcdot@[2]{\mathbin{\vcenter{\hbox{\scalebox{#2}{$\m@th#1\bullet$}}}}}
\makeatother

\usepackage{amsfonts}
\usepackage{amsmath}
\usepackage{amssymb}
\usepackage{url}
\usepackage{fancyhdr}
\usepackage{indentfirst}
\usepackage{enumerate}
\usepackage{titlesec}
\usepackage{amsthm}
\usepackage{dsfont}
\usepackage[misc]{ifsym}
\theoremstyle{definition}

\newtheorem{assumption}{Assumption}

\theoremstyle{plain}
\newtheorem{theorem}{Theorem}
\newtheorem{lemma}{Lemma}
\newtheorem{proposition}{Proposition}
\newtheorem{corollary}{Corollary}
\theoremstyle{remark}
\newtheorem{remark}{Remark}

\usepackage{cases}

\theoremstyle{definition}

\def\N{\mathbb{N}}

\def\p{\mathbb{P}}
\def\E{\mathbb{E}}

\def\R{\mathbb{R}}

\def\M{\mathcal{M}}

\def\H{\mathcal{H}}

\def\X{\mathcal{X}}

\def\X{\mathcal{X}}

\def\d{\,\mathrm{d}}
\DeclareMathOperator*{\esssup}{ess\text{-}sup}
\DeclareMathOperator*{\essinf}{ess\text{-}inf}

\renewcommand{\[}{\left[}

\usepackage[onehalfspacing]{setspace}

\usepackage{bm}
\usepackage{tikz-qtree}
\usepackage{tikz}

\def\id{\mathds{1}}

\title{Probabilistic risk aversion for generalized rank-dependent functions}
\author{Ruodu Wang\thanks{Department of Statistics and Actuarial Science, University of Waterloo,  Canada. \Letter~\url{wang@uwaterloo.ca}}
\and 
Qinyu Wu\thanks{Corresponding author. Department of Statistics and Actuarial Science, University of Waterloo,  Canada. \Letter~\url{q35wu@uwaterloo.ca}}
}
\date{\today}

\begin{document}
	\maketitle
	\begin{abstract}

Probabilistic risk aversion, defined through 
quasi-convexity in probabilistic mixtures, is a common  useful property in decision analysis. We study a general class of non-monotone mappings, called the generalized rank-dependent functions, which includes the preference models of expected utilities, dual utilities, and rank-dependent utilities as special cases, as well as signed Choquet functions used in risk management. 
Our results fully characterize probabilistic risk aversion for generalized rank-dependent functions: This property is determined by the distortion function, which is precisely one of the two cases: those that are convex and those that  correspond to scaled quantile-spread mixtures. Our result also leads to seven equivalent conditions for quasi-convexity in probabilistic mixtures of dual utilities and signed Choquet functions.
As a consequence, although probabilistic risk aversion is   quite different from the classic notion of strong risk aversion for generalized rank-dependent functions,
these two notions coincide for dual utilities under an additional continuity assumption. 

%We also illustrate a conflict between convexity in mixtures and convexity in risk pooling among constant-additive mappings.

	\end{abstract}

 \noindent\rule{\textwidth}{0.5pt}

\section{Introduction}

Expected utility theory (\cite{vNM47}), dual utility theory (\cite{Y87}),  and rank-dependent utility theory (\cite{Q82})  are among the most popular probabilistic preference models, and they are closely related to several large classes of  law-based risk measures (\cite{MFE15,FS16}).  

These decision models and risk measures can be equivalently formulated on either a set of distributions   or a set of random variables. 
A popular operation on a set of distributions is a probabilistic mixture; for instance, the independence axiom of \cite{vNM47} is formulated using probabilistic mixtures.  Quasi-convexity in probabilistic mixtures is a  useful property in decision models and risk measures, and it means that 
if the first distribution is preferred over the second one, then the first distribution is also preferred to a mixture of the two. 
The property is called \emph{probabilistic risk aversion} by \cite{W94} to distinguish it from other notions of risk aversion such as weak risk aversion or strong risk aversion in the sense of \cite{RS70}.\footnote{Throughout the paper, we will mostly say ``quasi-convexity" instead of ``probabilistic risk aversion" to emphasize the mathematical essence of this property and to contrast it with other properties.}
In the context of optimal decision under ambiguity,  e.g., \cite{GS89}, this property is convenient for applying minimax theorems,   allowing us to exchange the order of 
a maximum (representing an optimal action) and an infimum (representing a worst-case probability measure) under mild conditions. 
The property of having both quasi-convexity and quasi-concavity   is called betweenness (\cite{D86}), which was introduced to weaken the independence axiom;  see \cite{W94} for the importance of these properties in decision theory. 
In mathematical finance, betweenness for risk measures corresponds to the property of convex level sets, as studied by  \cite{W06}, \cite{Z16} and \cite{WW20}.

It is well known that the expected utility model is linear with respect to probabilistic mixtures, thus both convex and concave (we omit ``in probabilistic mixtures" below unless it is contrasted to another sense of convexity). 
Whereas convexity is well understood for dual utilities and rank-dependent utility models (e.g., \cite{W94}), quasi-convexity is not completely characterized for these models, and the result remains unknown even if under an increasing monotonicity (in the weak sense) or continuity condition. Although being weaker, quasi-convexity is similar to convexity, and 
as far as we know, the equivalence results between quasi-convexity and convexity in the literature are all under a strict monotonicity condition (see \cite{W94}; \cite{WY21}). 
%\cite{W94} showed that 
%they are equivalent for dual utilities, and \cite{WY21} generalized the equivalence result on the space which contains at least three nonindifferent outcomes.  
Nevertheless, if we remove this strict monotonicity, there are commonly used functionals in decision theory, statistics and risk management, such as left and right quantiles  which are quasi-convex but not convex. 

The main aim of this paper is to understand  quasi-convexity for a large class of mappings, called the \emph{generalized rank-dependent functions}, which include dual utilities and rank-dependent utilities as special cases.  
Our main result is a full characterization of quasi-convexity for this class with a very weak assumption on the domain of the functions, that is, the domain contains all distributions supported on
at least three nonindifferent outcomes.
This characterization is built on a corresponding result on \emph{signed Choquet functions}, a class of non-monotone and law-based mappings studied recently by \cite{WWW20a,WWW20}, and the corresponding mappings without law-basedness were investigated earlier in \cite{S86}.

Although most preference models in decision theory are monotone (either with respect to some notions of stochastic dominance or other orders),
there are three main advantages of working with non-monotone mappings, justifying the relevance of the study in this paper.    
First, signed Choquet functions include many popular non-monotone objects in risk management, such as the mean-median deviation, the Gini deviation,   the inter-quantile range, and the inter-Expected Shortfall range; see the examples in \cite{WWW20}. Note that variability measures in the sense of \cite{FWZ17} are never monotone with respect to first-order stochastic dominance.  
Second, removing monotonicity from the analysis allows us to have a deeper understanding of the essence of   important properties, such as quasi-convexity, by disentangling monotonicity from them. 
The third advantage concerns technical convenience and unification. With monotonicity relaxed, results on convexity and concavity, or those on maxima and minima, are symmetric; we only need to analyze one of them, and the other is clear automatically. This is particularly helpful when we switch between the world of risk measures (a smaller value is preferred) and that of utilities (a larger value is preferred).  
That being said, it is not our intention to argue against   monotonicity  in decision making; opening up the discussions on non-monotone mappings indeed helps to better understand monotone ones.  %All results in this paper are new with or without the additional assumption of monotonicity.  

%To fully characterize quasi-convexity for generalized rank-dependent functions, 
We begin by collecting  definitions and some preliminaries in Section \ref{sec:2}.
 In Section \ref{sec:RDU}, we focus on two important models in decision theory, dual utilities and rank-dependent utilities, by presenting a full characterization of their quasi-convexity (Theorem \ref{th-mainRDU}).
 This result implies, in particular, that for a dual utility with a continuous distortion function, strong risk aversion in the sense of \cite{RS70} is equivalent to probabilistic risk aversion. 
  We discover a new risk functional, called \emph{min-quantile mixture}, as the only possible form of dual utilities which are quasi-convex, other than the ones with convex distortion functions. To highlight the class of min-quantile mixtures, we use some properties to pin down it (Proposition \ref{th:mqm}).
Our main technical result (Theorem \ref{th-main})  in Section \ref{sec:mqc}   establishes a characterization of all quasi-convex generalized rank-dependent functions,   more general than those treated in Theorem \ref{th-mainRDU}.  
The characterization only depends on the distortion functions.
The class turns out to contain slightly more than those with convex distortion functions:  
A signed Choquet function is quasi-convex in probabilistic mixtures if and only if it is either convex in probabilistic mixtures or it is a \emph{scaled quantile-spread mixture} (more general than min-quantile mixtures). Based on our main result, a unifying equivalence result on signed Choquet functions (Theorem \ref{th-grand}) is presented:
%(without a proof):
If a distortion function is continuous, then quasi-convexity  is equivalent to six other equivalent conditions.
In Section \ref{sec:proof}, we give   proofs of our main result which relies on delicate technical analysis and use some results of \cite{DK82}, \cite{W94}, \cite{WWW20} and \cite{WW20}.  The technical challenges may explain why the result was not available before, given the prominence of both concepts of quasi-convexity and rank-dependent utilities. Some implications of Theorem \ref{th-main} for quasi-concavity and quasi-linearity are also reported in this section. In Section \ref{app:A}, we present a conflict between convexity in probabilistic mixtures and convexity in risk pooling among the class of constant-additive mappings. 
Section \ref{sec:7} concludes the paper.

\section{Preliminaries}
\label{sec:2}

In this section, we present some background on convexity, quasi-convexity, rank-dependent utilities,  and generalized rank-dependent functions.

\subsection{Convexity and  quasi-convexity}

%We first explain the key concepts in the paper, concavity and  quasi-concavity in probabilistic mixtures.
In this paper, the term ``distribution" represents a probability measure over a set of outcomes which is the real line $\R$. 
Let $\mathcal M$ be a set of distributions, and we always assume that it is  convex throughout the paper.
%Let $(\Omega,\mathcal F,\p)$ be a probability space. In this paper, distributions are always represented by the cumulative distribution functions of random variables on $(\Omega,\mathcal F,\p)$, i.e, for a random variable $X$, the distribution of $X$ is 
%Let $\mathcal M$ be a set of distributions, and we assume that it is always convex throughout the paper.
A mapping $\rho:\M\to \R$ is \emph{p-convex}  if 
$$
\rho(\lambda F+(1-\lambda)G) \le \lambda \rho(F) +(1-\lambda) \rho(G) \mbox{~~~for all $F,  G\in \M$ and $\lambda \in [0,1]$},
$$
and it is \emph{p-quasi-convex} if 
$$
\rho(\lambda F+(1-\lambda)  G) \le \max\{\rho(F), \rho(G) \} \mbox{~~~for all $F,G\in \M$ and $\lambda \in [0,1]$.}
$$ 
As usual, p-concavity and p-quasi-concavity  
are defined
by using
$\rho(\lambda F+(1-\lambda)G) \ge \lambda \rho(F) +(1-\lambda) \rho(G) $
and $\rho(\lambda F+(1-\lambda)  G) \ge \min\{\rho(F), \rho(G) \}$, respectively,  in the formulation above.
 \emph{P-quasi-linearity} of a functional $\rho$ means that it is both p-quasi-convex and p-quasi-concave.
The reason that we emphasize ``p"  (which stands for ``probabilistic") for these properties will be explained soon, as another form of convexity and concavity will appear and be contrasted.

Quasi-convexity is an ordinal property, whereas convexity is not. Indeed, 
for a preference $\succeq$ on $\M$ numerically represented by $\rho$, i.e., $F\succeq G \Longleftrightarrow \rho(F)\ge \rho(G)$,
quasi-convexity of $\rho$ corresponds to the following  property of $\succeq$ (see e.g., \cite{WY19, WY21}),
$$
\Big (F  \succeq  G  ~\Longrightarrow~  F \succeq  \lambda F +(1-\lambda)G \Big)\mbox{~~~for all $F,G\in \M$ and $\lambda \in [0,1]$.}
$$
This property is known as  {probabilistic risk aversion} by \cite{W94}.
Intuitively, it means that the decision maker with preference $\succeq$
dislikes combining two equally favourable distributions via a  random draw, 
which generally induces additional randomness.

Convexity or concavity of a mapping $\rho$  is commonly used in risk management, where the value of $\rho$, typically representing a monetary value, is primitive; see e.g., \cite{FS16}.   Quasi-convexity or quasi-concavity of $\rho$ is commonly used in decision theory, where the preference relation $\succeq$ is the primitive; see e.g., \cite{W10} and \cite{CMMM11}.
To unify both literature, we will formulate all properties on the numerical representation $\rho$.

Let $\X$ be a set of random variables in a fixed probability space $(\Omega,\mathcal F,\p)$ such that, first, $\X$ is \emph{law-based}, that is, if $X\in\X$ and $Y$ has the same distribution of $X$, then $Y\in\X$, and second, $\mathcal M=\{\p\circ X^{-1}: X\in\X\}$, where $X^{-1}$ is the set-valued inverse of $X$. That is, the set of distributions of all random variables in $\X$ is exactly $\M$. 
To guarantee the existence of $\X$ for all $\M$, we assume that
the probability space is nonatomic.\footnote{A probability space $(\Omega,\mathcal F,\p)$ is \emph{nonatomic} if for each $A\in\mathcal F$
	with $\p\left(  A\right)  >0$ there exists $B\in\mathcal F$ contained in $A$ such
	that $0<\p\left(  B\right)  <\p\left(  A\right) $.}
%Such set $\mathcal X$ always exists if the probability space is nonatomic.
%Let $\M_c$ be the set of compactly supported distributions on $\R$.
A mapping $\rho$ from $\M$ to $\R$ can be equivalently formulated as a mapping $\rho$ from $\X$ to $\R$ via $\rho(X):=\rho(F)$ where $F$ is the distribution of $X$; here we slightly abuse the notation to use   $\rho$ to represent both, and this should   be clear from the context. Such a mapping  $\rho$ on $\X$ is \emph{law-based}; that is, if $X,Y$ have the same distribution, then $\rho(X)=\rho(Y)$.
%where $\laweq$ represents equality in distribution.\footnote{In decision theory, this property is called probabilistic sophistication by \cite{MS92}.}   
We need both versions of the same mapping to make some interesting contrasts. 
When $\mathcal X$ is a convex set,
a mapping $\rho$ is \emph{o-convex} (where ``o" stands for ``outcome") on $\X$ if
\begin{align*}
	\rho(\lambda X+(1-\lambda)Y) \le \lambda \rho(X) +(1-\lambda) \rho(Y) \mbox{~~~for all $X,  Y\in \mathcal X$ and $\lambda \in [0,1]$},
\end{align*}
and it is \emph{o-quasi-convex} if 
\begin{align*}
	\rho(\lambda X+(1-\lambda) Y) \le \max\{\rho(X), \rho(Y) \} \mbox{~~~for all $X,Y\in \mathcal X$ and $\lambda \in [0,1]$}.
\end{align*}
O-concavity and o-quasi-concavity are defined similarly.  These properties are common for risk measures (\cite{ADEH99,FS16}). Moreover, \cite{DBW24} extended the study of o-quasi-convexity to the space of sequences of bounded random variables, further illustrating the relevance of these concepts in various frameworks.

For the same mapping $\rho$,  o-convexity 
and p-convexity have different interpretations, and sometimes they conflict with each other.
For instance, the variance is \emph{o-convex and  p-concave}.
This conflict is indeed intuitive, because a mixture of  two random losses, representing diversification,  reduces variability, whereas a mixture of two  distributions, representing   throwing a die to determine between two models, increases variability.\footnote{In the risk measure literature, o-quasi-convexity on $\X$ is  argued by \cite{CMMM11} to better represent the consideration of diversification.}
A mapping may be both o-convex and p-convex, and an example is the expected utility, $F\mapsto  \int u \d F$ for a convex function $u$; this mapping   is indeed o-convex and p-linear (i.e., both p-convex and p-concave). Nevertheless, Proposition \ref{pr:conflict} in Section \ref{app:A} shows a conflict between o-convexity and p-convexity; that is, among continuous and constant-additive mappings (e.g., monetary risk measures of \cite{FS16}), only a scaled expected value satisfies both properties.

We collect some notation used later in the paper. Denote by $\M_c$ the set of all compactly supported distributions on $\R$, and $\X_c$, which is defined on a nonatomic probability space, is the set of all random variables having distributions in $\M_c$. Note that both $\M_c$ and $\X_c$ are convex. We use $\esssup X$ and $\essinf X$ to represent the essential supremum and the essential infimum of a random variable $X$ on a probability space, respectively.
Denote by $\delta_x$ the point-mass at $x\in\R$. The function $\id_{A}$ is the indicator function of an event $A$.  Throughout, terms like ``increasing" and ``decreasing" are in the non-strict sense. All real-valued functions are tacitly assumed to be  measurable. We say that a real-valued function is (strictly) monotone if it is (strictly) increasing or (strictly) decreasing. A functional $\rho:\mathcal M\to\R$ is monotone if $\rho(F)\le \rho(G)$ for all $F,G\in\mathcal M$ such that $F\le_{\rm FSD} G$ where $\le_{\rm FSD}$ represents the first-order stochastic dominance, i.e., $F\le_{\rm FSD} G$ means that $\int f\d F\le \int f\d G$ for all increasing $f:\R\to\R$.

\subsection{Generalized rank-dependent functions}
\label{sec:22}
%%Let 
% $\M\subseteq\mathcal M_c$ be a set of distributions on $\R$.
 We first formulate signed Choquet functions and dual utilities, and then  introduce rank-dependent utilities and generalized rank-dependent functions. 
%We assume $\M\subseteq \M_c$ which ensures finiteness of these mappings. \label{added explanation.}
Signed Choquet functions are law-based mappings which are additive for comonotonic random variables (Theorem 1  of \cite{WWW20}, based on Proposition 2 of \cite{S86}), but not necessarily monotone.  
Denote  the set of all \emph{distortion functions} by $ \H ^{\rm BV}$,  $$\H^{\rm BV}=\{h:  [0,1]\to\R \mid \mbox{$h$ is of bounded variation and }h(0)=0 \},$$
 and  its subset of increasing normalized distortion functions by $ \H^{\rm DT}$,\footnote{DT stands for the dual theory of choice under risk of \cite{Y87}.}
 $$\H^{\rm DT}=\{h:  [0,1]\to\R  \mid \mbox{$h$ is increasing, $h(0)=0$ and $h(1)=1$} \}.$$   
% Throughout, terms like ``increasing" and ``decreasing" are in the non-strict sense.
A \emph{signed Choquet function} $I_h: \M \rightarrow \R$ is defined as  
\begin{align}
I_h(F)=\int_0^\infty h\circ F((x,\infty))\d x + \int_{-\infty}^{0} (h\circ F((x,\infty))-h(1) )\d x,
\label{eq:1}
\end{align}
where $h\in \H^{\rm BV}$ is its distortion function.  If $h\in \H^{\rm DT}$, then $I_h$ is called a \emph{dual utility} of \cite{Y87}. 
As explained earlier, $I_h$ is also formulated on $\X$, the set of random variables which have distributions in $\mathcal M$, via 
\begin{align}
I_h(X)=\int_0^\infty h(\p(X>x))\d x + \int_{-\infty}^{0} (h(\p(X>x))-h(1) )\d x.
\label{eq:1-prime}
\end{align}  
%The mapping $I_h$ can be defined on domains other than $\M_c$ and $\X_c$ via  \eqref{eq:1} and \eqref{eq:1-prime} whenever the integrals are finite.  
%In some results, we treat $I_h$ as a mapping from $\M$ to $\R$ for a general $\M$. 

As a main subject of the paper, 
a \emph{generalized rank-dependent function}
is a mapping  $R_{h,v}:\M\to \R$ defined by, for  $h\in \mathcal H^{\rm BV}$ and $v:\R\to \R$,
\begin{align}\label{eq-GRDU}
	R_{h,v}(F)=\int_0^\infty h\circ F\circ v^{-1}((x,\infty))\d x + \int_{-\infty}^{0} (h\circ F\circ v^{-1}((x,\infty))-h(1) )\d x,
\end{align}
where $v^{-1}$ is the set-valued inverse of $v$.
The function $v$ is typically increasing in economic applications (e.g., a utility function). The signed Choquet function is a special case of $R_{h,v}$ with $v$ being the identity, and
if $h\in \mathcal H^{\rm DT}$ and $v$ is increasing, then $R_{h,v}$ corresponds to a \emph{rank-dependent utility} of \cite{Q93}. 
Although rank-dependent utilities are well studied in decision theory (see e.g., \cite{WZ19,EL22}), the class of  generalized rank-dependent functions is newly introduced in this paper.
 
A $v$-transform
maps the distribution of a random variable $X$ to the distribution of $v(X)$.\footnote{As shown by \cite{LSW21}, commutation with $v$-transforms essentially characterizes probability distortions among all distributional transforms.} In other words, the distribution $F$ is transformed to $F\circ v^{-1}$.
The mapping $R_{h,v}$ can be formulated as a signed Choquet function under a $v$-transform, i.e.,
\begin{align*}
		R_{h,v}(F)=I_h(F\circ v^{-1})
\end{align*}
or equivalently formulated on $\X$ via $ R_{h,v}(X) = I_{h}(v(X))$. 
%In economics and finance, $v$-transforms are widely used. 
In addition to the expected utility or the rank-dependent utility models,
a pricing functional for options can be seen as the (risk-neutral) expectation of a transformed asset price distribution.

We will encounter discrete distributions throughout the paper. An explicit representation of $R_{h,v}$ with discrete distributions is given below.
%Denote by $\delta_x$ the point-mass at $x\in\R$. 
For a discrete distribution $F$ with the form $F=\sum_{i=1}^n p_i\delta_{x_i}$ where $v(x_1)\ge  \dots\ge v(x_n)$, $p_1,\dots,p_n\ge 0$ and $\sum_{i=1}^n p_i=1$, it holds that 
$$ 
R_{h,v}(F)=\sum_{i=1}^n (h(p_1+\cdots+p_i)-h(p_1+\cdots+p_{i-1}))v(x_i).
$$
The following assumption on the %function $v$ and the 
set of distributions $\M$ will be useful for our characterization results. 

\renewcommand\theassumption{M}
\begin{assumption}\label{assm:1}
%The function $v:\R\to\R$ in the generalized rank-dependent function $R_{h,v}$ satisfies that there exist $x,y,z\in\R$ such that $v(x)>v(y)>v(z)$. 
The set $\M$ is a convex subset of $\mathcal M_c$ and there exist three distinct points $x,y,z\in\R$ such that $\delta_x,\delta_y, \delta_z \in \M$.
%and contains all three-point distributions on three distinct points $x>y>z$, that is,
%\begin{align}\label{eq-setassumption0}
%\left\{p\delta_x+(q-p)\delta_y+(1-q)\delta_z: 0\le p\le q\le1\right\}\subseteq \M.
%\end{align}
\end{assumption}
Assumption \ref{assm:1} is very weak and harmless for any practical purpose. 
For a set $\M$ satisfying Assumption \ref{assm:1}, denote by 
\begin{align}\label{eq-setv}
	\mathcal V_{\M}=\{v:\R\to\R \mid \mbox{$v(x)$, $v(y)$ and $v(z)$ are distinct for some $\delta_x,\delta_y,\delta_z\in \M$}\}.
\end{align}
In the set $\mathcal V_{\M}$,
we do not impose the continuity or the monotonicity on $v$, and we only require that $v$ can take at least three distinct values on the points $ x,y,z $. % in \eqref{eq-setassumption0}.
For instance, this requirement holds true if $v$ is strictly increasing.

%We do not impose the continuity or monotonicity on $v$, and we only require that $v$ can take at least three distinct values.
%For the characterization of generalized rank-dependent functions in Theorem \ref{th-main}, one may safely take $\M=\M_c$ when $v$ takes at least three distinct values.

%The mapping $I_h$ will be restricted to a domain $\M$ that is possibly smaller than $\M_c$, and all properties of $I_h$ are confined to the domain.  
%Furthermore, denote  by  $\mathcal F\subseteq\mathcal M_c$ a set of distributions that contains at least all three point distributions with fixed values $x>y>z$, i.e., %\begin{align}\label{eq-setassumption}
%\left\{p\delta_x+(q-p)\delta_y+(1-q)\delta_z: 0\le p\le q\le1\right\}\subseteq \mathcal F\subseteq\mathcal M_c.
%\end{align}

In what follows, we will take the convexity of $h$ as the primary property as opposite to concavity to consider,
as the preference represented by $I_h$ is strongly risk averse in the sense of \cite{RS70} if $h$ is increasing and convex (Theorem 2 of \cite{Y87}). Changing convexity to concavity makes no real mathematical difference since all results can be written for concavity via a sign change; recall that this is an advantage of working with non-monotone mappings such as generalized rank-dependent functions.

\section{P-quasi-convexity  of  rank-dependent utilities}\label{sec:RDU}

Given the importance of rank-dependent utilities  in economics and finance,
we first present results for this class, although these results find their more general versions in Section \ref{sec:mqc}. 

\subsection{P-quasi-convexity, concavity and linearity}
%We first focus on dual utilities and rank-dependent utilities, which are two classic models in economics and finance. 
Recall that rank-dependent utility is defined by  \eqref{eq-GRDU}
where $h\in\mathcal H^{\rm DT}$ and $v:\R\to\R$ is increasing. The dual utility is a special case of $R_{h,v}$ when $h\in\mathcal H^{\rm DT}$ and $v$ is the identity. %We will show the full characterization of $p$-quasi-convexity for dual utilities and rank-dependent utilities.

%At first, 
We first introduce a few special cases of dual utilities and distortion functions in $\mathcal H^{\rm DT}$ that are important for understanding p-quasi-convexity. There is a one-to-one correspondence between dual utilities and distortion functions in $\mathcal H^{\rm DT}$, and hence, it suffices to study the distortion functions.
For  a distribution $F $, the \emph{left-} and \emph{right-quantile} at level $\alpha\in(0,1)$ are respectively defined by
\begin{align*}
	Q^{-}_{\alpha}(F)=\inf\{x\in \R: F(x)\ge \alpha\}~~{\rm and}~~Q^{+}_{\alpha}(F)=\inf\{x\in \R : F(x)> \alpha\}.
\end{align*} 
For using quantiles as preferences in decision theory, see \cite{R10}.
At the level $\alpha=0$ or $1$, left- and right-quantiles coincide, and they are defined by
\begin{align*}
	Q_0^+(F)&=Q_0^-(F)=Q_0(F)=\inf\{x\in \R : F(x)> 0\}; \\
	Q_1^+(F)&=Q_1^-(F)=Q_{1}(F)=\inf\{x\in \R : F(x)\ge 1\}.
\end{align*}
For some $c\in[0,1]$ and $\alpha\in[0,1]$, the \emph{mixed quantile} is defined by $Q_\alpha^{c}=cQ_\alpha^{+}+(1-c)Q_\alpha^{-}$.
%All mixed quantiles are p-quasi-convex and p-quasi-concave. 
All mixed quantiles have \emph{convex level sets} (CxLS), i.e., $\rho(F)=\rho(G)$ $\Longrightarrow$ $\rho(\lambda F+(1-\lambda)G)=\rho(F)$ for all $\lambda \in [0,1]$ and $F,G\in \M$. This property and monotonicity together imply   p-quasi-convexity and p-quasi-concavity.\footnote{All signed Choquet functions with convex level sets are characterized by \cite{WW20}, which are slightly more than linear transformations of the quantiles and the mean.}
Finally, we introduce a new class of functionals, called the \emph{min-quantile mixtures}, defined by
\begin{align}\label{eq-qmm}
	k Q_\alpha^c+(1-k) Q_0,~~~~ \mbox{for some~} \alpha,c,k\in [0,1].
\end{align}
A min-quantile mixture is a convex combination of the essential infimum and a mixed quantile. \begin{figure}[t]%[htbp]
	{  \setlength{\unitlength}{2.5cm}
		\begin{center}%\footnotesize
			\begin{picture}(1.8,1.5)(0,0) 
				\put(0,0){\vector(0,1){1.3}}
				\put(0,0){\vector(1,0){2}}
				\put(1.73,-0.17){$ 1$}
				\put(-0.15,-0.15){$0$}
			%	\put(-0.3,-0.55){$-b$}
				\put(-0.3,0.15){  $kc$}
				\put(-0.27,0.45){ $k$}
				\put(-0.25,0.9){ $1$}
				\put(0.82,-0.17){  $\alpha$}
				\put(2.05,-0.1){$ p$}
				\put(0.05,1.22){$ h(p)$}  
				\linethickness{0.3mm}
		%		\qbezier(0.035,-0.5)(0.2,-0.5)(0.89,-0.5)
				\qbezier(0.935,0.5)(1.5,0.5)(1.73,0.5) 
				\linethickness{0.3mm} 
				\qbezier(0,0)(0.5,0)(0.885,0) 
				\linethickness{0.1mm} 
				\multiput(0.0,0.2)(0.1,0){9}{\line(1,0){0.05}}
				\multiput(0.0,0.95)(0.1,0){18}{\line(1,0){0.05}} 
%				\put(-0.095,-0.04) {  $  \bullet $}
		%		\put(-0.105,-0.55){ $   \circ  $}
				
	%			\put(0.82,-0.55){   $ \circ $}	
	            \put(0.82,-0.04){ $ \circ $}
				\put(0.82,0.16){ $\bullet$}
				
				\put(0.82,0.465){ $ \circ $}
				\put(1.67,0.91){  $ \bullet $}
				\put(1.67,0.465){ $ \circ $}
			\end{picture}
		\end{center}
	}
	\caption{The distortion function of a min-quantile mixture.}\label{fig:qmm}
\end{figure}  The distortion function of a min-quantile mixture has the form (see Figure \ref{fig:qmm}):
\begin{align*}
h(p)=kc\id_{\{p=\alpha\}}+k\id_{\{\alpha<p<1\}}+\id_{\{p=1\}},~~p\in[0,1].
\end{align*} 

For strictly increasing distortion functions $h$, \cite{W94} and \cite{WY21}
showed that only the convex ones are possible for $I_h$ to be p-quasi-convex. In the larger class  of $I_h$ with $h\in\mathcal H^{\rm DT}$, the following result illustrates that 
the min-quantile mixtures are the only other choice 
satisfying p-quasi-convexity besides those   with convex distortion functions.

\begin{theorem}\label{th-mainRDU}
Suppose that   Assumption \ref{assm:1} holds, $h\in\mathcal H^{\rm DT}$, and  $v:\R\to\R$ is increasing. The following statements are equivalent.
\begin{itemize}
		\item[(i)] $h$ is convex or $I_h$ is a min-quantile  mixture.
		%$=kQ_{\alpha}^c+(1-k)Q_0$ for some $k,\alpha,c\in[0,1]$.
		\item[(ii)] $I_h$ is p-quasi-convex on $\mathcal M$.
		\item[(iii)] $R_{h,v}$ is p-quasi-convex on $\mathcal M$ for some $v\in \mathcal V_{\M}$.
		\item[(iv)] $R_{h,v}$ is p-quasi-convex on $\mathcal M$ for all  functions $v$.
\end{itemize}
\end{theorem}

The proof of the above theorem is follows from the more general result in Theorem \ref{th-main}  in the next section, which studies all generalized rank-dependent functions. 
In (iii), the condition that $v$ can take three different values is essential. Note that  if $v$ is a constant function, then $R_{h,v}$ is p-quasi-convex regardless of $h$.
  Section \ref{sec:mqc} has more discussions on the role of this condition.

 Next, we connect two notions of risk aversion.  
A functional $\rho:\M\to \R$ is \emph{concave-order monotone}
if   $\rho(F)\le \rho(G)$ for all $F,G\in \M_c$ such that $F\le_{\rm cv} G$ where $\le_{\rm cv}$ represents concave order between distributions, i.e., $F\le_{\rm cv}G$ means that $ \int f \d F \le \int f\d G$ for all concave $f:\R\to \R$. 
Concave-order monotonicity of  $\rho$ is equivalent to \emph{strong risk aversion} of the preference represented by $\rho$ in the sense of  \cite{RS70}. 

Since a min-quantile mixture does not have a continuous distortion function, 
  Theorem \ref{th-mainRDU} implies that the only class of continuous distortion functions
yielding p-quasi-convexity
   is that of the convex ones.
  \cite{Y87} showed that a dual utility is strongly risk averse   if and only if the distortion function is convex. 
The next corollary directly follows from this and Theorem \ref{th-mainRDU}.
  \begin{corollary}\label{cor:p-DT}
For a dual utility on $\M_c$ with a continuous distortion function,  probabilistic risk aversion is equivalent to  strong risk aversion.
  \end{corollary}
  
 \cite{W94} showed that the  conclusion of Corollary \ref{cor:p-DT}  holds for strictly increasing $h$. 
Continuity  of $h$ is quite natural, which is implied by the axioms of \cite{Y87}.
 From a decision theoretical standpoint, a continuous distortion function $h$ is empirically plausible as in, e.g., Equation (6) of \cite{TK92} in the framework of cumulative prospect theory.
For rank-dependent utilities, strong risk aversion further requires $v$ to be concave (\cite{CKS87}), and hence the conclusion of Corollary \ref{cor:p-DT} does not hold. 

The following characterizations of p-quasi-concavity and p-quasi-linearity can be obtained in a similar way  to Theorem \ref{th-mainRDU}. Their more general versions are  Corollaries \ref{co-QCX} and \ref{co-linearity} in Section \ref{sec:proof}.
\begin{proposition}\label{prop-DT}
	Suppose that  Assumption \ref{assm:1} holds, and let $v\in\mathcal V_{\M}$ be increasing and  $h\in\mathcal H^{\rm DT}$.
	\begin{itemize}
%		\item[(i)] The following are equivalent: $R_{h,v}$ is p-quasi-convex; $I_h$ is p-quasi-convex;  $h\in\mathcal H_{\rm DT}^{\rm QCX}$.
		\item[(i)] The following are equivalent: $R_{h,v}$ is p-quasi-concave; $I_h$ is p-quasi-concave; $h$ is concave or $I_h=kQ_{\alpha}^c+(1-k)Q_1$ for some $k,\alpha,c\in[0,1]$.
		\item[(ii)] The following are equivalent: $R_{h,v}$ is p-quasi-linear; $I_h$ is p-quasi-linear; $I_h$ is one of the forms: 
		$I_h=\E$; $I_h=cQ_1+(1-c)Q_0$ for some $c\in[0,1]$; $I_h=Q^c_{\alpha}$ for some $c\in[0,1]$ and $\alpha\in(0,1)$.
	\end{itemize}
\end{proposition} 

\begin{remark}
	For $h\in\mathcal H^{\rm DT}$, $I_h$ is monotone and translation invariant.\footnote{A functional
%		 $\rho: \X \to\R$ is monotone if $\rho(X)\le \rho(Y)$ for all $X,Y\in \X$ satisfying $X\le Y$ $\p$-a.s., and it 
		 is translation invariant if $\rho(X+c)=\rho(X)+c$ for all $X\in \X$ and $c\in\R$. }  By Lemma 2.2 of \cite{BB15}, $I_h$ is p-quasi-linearity if and only if it has CxLS. 
		 The class of dual utilities $I_h$ satisfying CxLS is characterized by \citet[Theorem 2]{KP16}, which are those in Proposition \ref{prop-DT} (ii). 
\end{remark}

\subsection{The min-quantile mixtures}

As the min-quantile mixtures are the
only dual utilities satisfying p-quasi-convexity among  besides   convex ones, we can find several properties that identity the min-quantile mixtures. 
For this, some terminology and properties are needed. We say that random variables $X$ and $Y$ are \emph{comonotonic} if there exists $\Omega_0\in\mathcal F$ with $\p(\Omega_0)=1$ such that for all $\omega,\omega'\in\Omega_0$,
\begin{align*}
	(X(\omega)-X(\omega'))(Y(\omega)-Y (\omega'))\ge 0.
\end{align*}
For a mapping $\rho:\mathcal X_c\to \R$ (also treated as a mapping from $\M_c$ to $\R$), we say that $\rho$ is \emph{comonotonic-additive} if for any comonotonic random variables $X,Y\in \mathcal X_c$, $\rho(X+Y)=\rho(X)+\rho(Y)$; %$\rho$ is \emph{uniformly norm-continuous} if it is uniformly continuous with respect to the norm defined by $\|X\|=\esssup|X|$; 
$\rho$ is \emph{p-locally indifferent} if
there exist  $A,B\in\mathcal F$ such that $\p(A)\neq\p(B)$ and $\rho(\id_A)=\rho(\id_B)>\rho(0)$.

Comonotonic additivity is popular in both the literature  of decision theory (e.g., \cite{Y87,S89}) and that of risk measures (e.g., \cite{K01}),
whereas p-local indifference is less known. 
Intuitively, p-local indifference  means that $\rho$ may be indifferent between two events $A$ and $B$ even if their probabilities are different. 
This property is called ``local" because it only states the existence of such a pair of events, but says nothing about general pairs.
Although it may be explained by the inability to assess the small difference between different probabilities, such a behaviour is empirically observed in behaviour experiments where only a few categories of probability are considered. For instance,  \cite{PB75}   showed that children of age 4--5 have only three levels of plausibility conception: certainly true, certainly untrue, uncertain. This is compatible with the min-quantile mixture whose distortion function has two jumps and is constant elsewhere. 
When the probability categories were related to verbal expressions, people often prefer to use only a few verbal categories of probability rather than a continuous scale; see \cite{W86}, \cite{WB95} and \cite{WW96}. \cite{CJS87} investigated the individual decision making under risk and under nonprobabilized uncertainty where they found that on the loss side, people appear to have coarser categories of plausibility than the continuum $[0,1]$.

The next result characterizes the min-quantile mixtures.

\begin{proposition}\label{th:mqm}
	A law-based functional $\rho: \mathcal X_c\to \R$ is monotone, comonotonic-additive, p-quasi-convex  and p-locally indifferent with $\rho(1)=1$ if and only if $\rho=kQ_\alpha^c+(1-k)Q_0$ for some $c\in[0,1]$ and $\alpha, k\in(0,1]$, that is, $\rho$ is a min-quantile mixture excluding the case $Q_0$.
\end{proposition}

\begin{proof}
	Let us first verify sufficiency.  Monotonicity, comonotonic-additivity and $\rho(1)=1$ are trivial. The p-quasi-convexity follows from Theorem \ref{th-mainRDU}. 
	It remains to show that $\rho$ is p-locally indifferent. Let $A$ and $B$ be such that $\p(A)=1-\alpha/2$ and $\p(A)=1-\alpha/3$. Then, we have $\rho(\id_A)=\rho(\id_B)=k>0$. This completes the proof of sufficiency. To see necessity, it follows from \citet[Theorem 1]{WWW20} that $\rho$ is a signed Choquet function, i.e., $\rho(X)=\int_0^\infty h(\p(X>x))\d x + \int_{-\infty}^{0} (h(\p(X>x))-h(1))\d x$ with some $h\in\mathcal H^{\rm BV}$. By \citet[Lemma 2]{WWW20}, monotonicity implies that $h$ is increasing. Furthermore, $\rho(1)=1$ means $h(1)=1$. Hence, we have concluded $h\in\mathcal H^{\rm DT}$. Applying Theorem \ref{th-mainRDU}, we have $\rho=I_h$ for some convex $h\in\mathcal H^{\rm DT}$ or $\rho=kQ_{\alpha}^c+(1-k)Q_0$ for some $k,\alpha,c\in[0,1]$. 
	Since $Q_0$ also has a convex distortion function $h(t)=\id_{\{t=1\}}$,  we only need to verify that $I_h$ with convex $h$  is not p-locally indifferent. 
	P-local indifference implies  $\rho(\id_A)=\rho(\id_B)$ for some $A$, $B$ with $0<\p(A)<\p(B)\le 1$. Denote by $p=\p(A)$ and $q=\p(B)$. It holds that $h(p)=h(q)>0$ with $0<p<q\le 1$, and this implies that $h$ cannot be convex. Hence, we complete the proof.
\end{proof}

We do not aim to promote the min-quantile mixtures as a plausible decision criterion or a new risk measures, as its interpretation and applications are  similar to quantiles.  Nevertheless, we find the class of min-quantile mixtures intriguing, both mathematically and decision-theoretically, as the only class within dual utilities that can accommodate both p-local indifference (inability to distinguish probabilistic assessments) 
and probabilistic risk aversion (dislike of randomness), highlighting a conflict between the two properties.

\section{Generalized rank-dependent functions}
\label{sec:mqc}

This section is the main part of the paper where we obtain a full characterization of p-quasi-convexity for generalized rank-dependent functions, as well as a unifying  equivalence result for signed Choquet functions with strictly monotone or continuous distortion functions. We will mention Assumption \ref{assm:1} if it is needed in a result.
\subsection{Invariance under transforms and three special distortion functions}

%This section is the main part of the paper where we investigate a full characterization of p-quasi-convexity of $R_{h,v}$.
%%which functions $h\in\mathcal H^{\rm BV}$ yield p-quasi-convexity of $R_{h,v}$.  
%%We assume $\M=\M_c$ in all results, and
%We will mention Assumption \ref{assm:1} if it is needed in a result.
  
Recall that the generalized rank-dependent function is a $v$-transform of a signed Choquet function.
Similarly to Proposition 3.5 of \cite{WW20}, we first show that p-quasi-convexity is an invariance property under $v$-transforms.

\begin{proposition}\label{prop-transform}
	For $v:\R\to\R$ and a convex set of distributions $\mathcal M$, define  $\mathcal M^v=\{F\circ v^{-1}: F\in\mathcal M\}$. Then, $R_{h,v}$
is p-quasi-convex on $\mathcal M$ if and only if $I_h$ is p-quasi-convex on $\mathcal M^v$.
\end{proposition}
\begin{proof}
	For any $\lambda\in(0,1)$ and $F,G\in\mathcal M$, denote by $F'=F\circ v^{-1} \in \M^v$,
	$G'=G\circ v^{-1}\in \M^v$, and  we have
	\begin{align*}
		(\lambda F+(1-\lambda)G)\circ v^{-1}=\lambda(F\circ v^{-1})+(1-\lambda)(G\circ v^{-1}) = \lambda F'+(1-\lambda) G', 
	\end{align*}
	which follows directly from the definition of probability measures. 
	Since
	\begin{align*}
		R_{h,v}(\lambda F+(1-\lambda)G) =I_h(\lambda F'+(1-\lambda) G') \mbox{~~and~~}
		\max\{R_{h,v}(F),R_{h,v}(G)\} =  \max\{I_h(F'),I_h(G')\},
	\end{align*}
	we obtain $$R_{h,v}(\lambda F+(1-\lambda)G)\le \max\{R_{h,v}(F),R_{h,v}(G)\}\mbox{~~$\Longleftrightarrow$~~}I_h(\lambda F'+(1-\lambda) G') \le  \max\{I_h(F'),I_h(G')\}. $$
	It follows that
	p-quasi-convexity of $I_h$ on $\M^v$ is equivalent to that of $R_{h,v}$ on $\M$, noting that  
	$F, G$ can be arbitrarily chosen from $\M$ and $F',G'$ can be arbitrarily chosen from $\M^v$. This completes the proof.
\end{proof}

The signed Choquet function $I_h$ in the above proposition can be replaced by any functionals $\rho$, and $\rho_v(F):=\rho(F\circ v^{-1})$ should take place of $R_{h,v}$ simultaneously.
To characterize p-quasi-convexity of generalized rank-dependent functions,
Proposition \ref{prop-transform} implies that we can alternatively consider the characterization of the same property in the class of signed Choquet functions. As a result, p-quasi-convexity of $R_{h,v}$ only depends on the distortion function $h$. 
More precisely, we will show below that the property of $h$ determines p-quasi-convexity of $R_{h,v}$ under Assumption \ref{assm:1} if $v\in\mathcal V_{\M}$
(Theorem \ref{th-main}), but the situation is different if $v$ can only take two values (Proposition \ref{prop:two-point}). It is trivial that p-quasi-convexity holds for all $h\in\mathcal H^{\rm BV}$ if $v$ is a constant function.

We define a few   special cases of signed Choquet functions that will appear later in our characterization result. 
%Note that 
%there is a one-to-one correspondence between signed Choquet functions and distortion functions, and hence, the distortion function can be identified based on the form of signed Choquet function. 
%All mappings below are defined on both $\M_c$ and $\X_c$, and we will use whichever is more convenient in different places. 
%The \emph{left-} and \emph{right-quantile} at level $\alpha\in(0,1)$ are respectively defined by
%\begin{align*}
%Q^{-}_{\alpha}(F)=\inf\{x: F(x)\ge \alpha\}~~{\rm and}~~Q^{+}_{\alpha}(F)=\inf\{x: F(x)> \alpha\}.
%\end{align*} 
%For using quantiles as preferences in decision theory, see  \cite{R10}.
%At the level $\alpha=0$ or $1$, left- and right-quantiles   coincide, and they are defined by
%\begin{align*}
%Q_0^+(F)=Q_0^-(F)=Q_0(F)=\inf\{x: F(x)> 0\}
%\end{align*}
%and
%\begin{align*}
%Q_1^+(F)=Q_1^-(F)=Q_{1}(F)=\inf\{x: F(x)\ge 1\}.
%\end{align*}
%For some $c\in[0,1]$ and $\alpha\in[0,1]$, the \emph{mixed quantile} is defined by $Q_\alpha^{c}=cQ_\alpha^{+}+(1-c)Q_\alpha^{-}$. 
An \emph{asymmetric (negative) spread} is the mapping  
$$ 
S_{a,b}= a Q_0-b Q_1,
$$ 
where $a,b \ge 0$. Note that  $-S_{1,1}=Q_1-Q_0$ is the usual spread of the support of the distribution. We omit ``negative" below for simplicity and   call $S_{a,b}$ an asymmetric spread.  
The \emph{scaled quantile-spread mixture} is given by  
\begin{align}
S_{a,b}+kQ_\alpha^c%=a Q_1 - b Q_0+ kQ_\alpha^c
,~~~~ \mbox{for some~} a,b\ge0, ~\alpha,c\in [0,1],~k\in \R. \label{eq-H2}
\end{align}
A scaled quantile-spread mixture is
 the sum of an asymmetric spread  $S_{a,b}$ and a scaled mixed quantile $kQ_\alpha^c$, and this class includes mixed quantiles, asymmetric spreads. and min-quantile mixtures as special cases.
 Note that we allow $\alpha=0$ or $\alpha=1$ in \eqref{eq-H2}, which leads to 
$
aQ_0- bQ_1$ where $a\ge 0$ or $b\ge 0$, but it does not include the case  $a,b<0$.
 Figure \ref{fig:distortion} reports an example of the distortion functions of an asymmetric spread, a scaled mixed quantile and a scaled quantile-spread mixture.

\begin{figure}[t]%[htbp]
	{  \setlength{\unitlength}{1.8cm}
		\begin{center}\footnotesize
			\begin{picture}(1.8,1.5)(0,-0.5)
				\put(0,-0.6){\vector(0,1){1.8}}
				\put(0,0){\vector(1,0){2}}
				\put(1.73,-0.17){$ 1$}
				\put(-0.15,-0.15){$0$}
				\put(-0.3,-0.55){$-b$}
				\put(-0.46,0.7){$ a-b$}
				\put(2.05,-0.1){$ p$}
				\put(0.05,1.22){$ h(p)$}  
				\linethickness{0.3mm}
				\qbezier(0.03,-0.5)(1.45,-0.5)(1.74,-0.5) 
				\linethickness{0.1mm} 
				\multiput(0.0,0.75)(0.1,0){18}{\line(1,0){0.05}}  
	            \put(-0.105,-0.035) {  $  \bullet $}
				\put(-0.105,-0.55){ $   \circ  $}
				\put(1.67,-0.55){ $ \circ $}
				\put(1.67,0.7){  $ \bullet $}
			\end{picture}\quad\quad\quad\quad\quad
			\begin{picture}(1.8,1.5)(0,0) 
				\put(0,0){\vector(0,1){1.8}}
				\put(0,0){\vector(1,0){2}}
				\put(1.73,-0.17){$ 1$}
				\put(-0.05,-0.17){$ 0$} 
				\put(-0.3,0.52){  $ kc$}
				\put(-0.2,1.43){   $ k$}
				\put(0.82,-0.17){   $\alpha$}
				\put(2,-0.1){   $ p$}
				\put(0,1.8){   $ h(p)$} 
				\linethickness{0.3mm}
				\qbezier(0,0)(0.6,0)(0.89,0)
				\qbezier(0.96,1.44)(1.8,1.44)(1.75,1.44) 
				\linethickness{0.1mm} 
				\multiput(0.0,0.55)(0.1,0){10}{\line(1,0){0.05}} 
				\put(0.82,-0.04){   $ \circ $}
				\put(0.82,0.51){  $ \bullet $}
				\put(0.82,1.4){   $ \circ $}
			\end{picture}\quad\quad\quad\quad\quad\quad\quad
			\begin{picture}(1.8,1.5)(0,-0.5) 
			\put(0,-0.6){\vector(0,1){1.8}}
			\put(0,0){\vector(1,0){2}}
			\put(1.73,-0.17){$ 1$}
			\put(-0.15,-0.15){$0$}
			\put(-0.3,-0.55){$-b$}
				\put(-0.75,0.15){  $-b+kc$}
				\put(-0.68,0.45){ $ -b+k$}
				\put(-0.84,0.9){ $a-b+k$}
				\put(0.82,-0.17){  $\alpha$}
	            \put(2.05,-0.1){$ p$}
                \put(0.05,1.22){$ h(p)$}  
				\linethickness{0.3mm}
\qbezier(0.035,-0.5)(0.2,-0.5)(0.89,-0.5)
\qbezier(0.96,0.5)(1.5,0.5)(1.74,0.5) 

\linethickness{0.1mm} 
\multiput(0.0,0.2)(0.1,0){9}{\line(1,0){0.05}}
\multiput(0.0,0.95)(0.1,0){18}{\line(1,0){0.05}} 
\put(-0.105,-0.035) {  $  \bullet $}
\put(-0.105,-0.55){ $   \circ  $}

\put(0.82,-0.55){   $ \circ $}
\put(0.82,0.15){  $ \bullet $}

\put(0.82,0.45){ $ \circ $}
\put(1.67,0.9){  $ \bullet $}
\put(1.67,0.45){ $ \circ $}
			\end{picture}
		\end{center}
}
	\caption{Distortion functions of an asymmetric spread (left), a scaled mixed quantile (middle) and scaled quantile-spread mixture (right) with $\alpha ,c\in (0,1)$ and $k>0$.}\label{fig:distortion}
\end{figure} 

We denote by  $\mathcal H^{\rm CX}$ the set of all convex distortion functions $h\in \mathcal H^{\rm BV}$ and  by $\mathcal H^{\rm QSM}$ the set of distortion functions of scaled quantile-spread mixtures. The   functions in $ \mathcal H^{\rm QSM}$ have the form
\begin{align}\label{eq:hQSM} h(p)=
-b\id_{\{0<p<\alpha\}}+(-b+kc)\id_{\{p=\alpha\}}
+(-b+k)\id_{\{\alpha<p<1\}}+(a-b+k)\id_{\{p=1\}} ,  ~p\in[0,1],
\end{align}
for some parameters in \eqref{eq-H2}. The asymmetric spread corresponds to $a,b>0$, $\alpha=1$ and $k=0$,   the mixed quantile corresponds to $a=b=0$, $\alpha,c\in(0,1)$ and $k=1$,
and a min-quantile mixture in \eqref{eq-qmm} corresponds to $b=0$, $k\in [ 0,1]$ and $a=1-k$. 
%
%The following two forms of $h$ in the subsets $\mathcal H^{\rm CV}$ and $\mathcal H^{\rm QSM}$ of $\mathcal H^{\rm BV}$  are important for p-quasi-concavity of $I_h$.
%\begin{itemize}
%\item[(i)] $h\in\mathcal H^{\rm CV}$: $h$ is concave. 
%\item[(ii)] $h\in\mathcal H^{\rm QSM}$: For some $a\ge 0$, $b\le 0$, $\alpha\in[0,1]$, $c\in[0,1]$ and $k\in\R$, $h(p)=
%a\id_{\{0<p<\alpha\}}+(kc+a)\id_{\{p=\alpha\}}
%+(k+a)\id_{\{\alpha<p<1\}}+(a+b+k)\id_{\{p=1\}}$, $p\in[0,1]$. In this case,
%\begin{align}\label{eq-H2}
%I_h=a Q_1+kQ_{1-\alpha}^{c}+b Q_0.
%\end{align}
%\end{itemize} 

%For strictly increasing distortion functions $h$, \cite{W94} and \cite{WY21}
%showed that only the convex ones are possible for $I_h$ to be p-quasi-convex.
%Without strict monotonicity, the class of signed Choquet functions is much larger, 
% and this includes the scaled quantile-spread mixtures. 
The next   lemma   
% is not only useful to verify the p-quasi-concavity of the scaled quantile-spread mixtures, but also will be used in the proof of the main theorem in the next section. 
% The lemma 
 says that the asymmetric spread $S_{a,b}$ can be safely added to any $I_h$ without changing its p-quasi-convexity, thus highlighting the special role of $S_{a,b}$.

 \begin{lemma}\label{lm-dis01}
 Let $\mathcal M$ be a convex set of distributions. Suppose that $h,\widetilde{h}\in\mathcal H^{\rm BV}$ satisfies $I_h=S_{a,b}+I_{\widetilde{h}}$ for some $a,b\ge 0$. If $I_{\widetilde{h}}$ is p-quasi-convex on $\mathcal M$, then $I_h$ is also p-quasi-convex on $\mathcal M$.
% 	\begin{itemize}
% 		\item[(i)] If $I_{\widetilde{h}}$ is p-quasi-concave, then $I_h$ is p-quasi-concave.
% 		\item[(ii)] If $\widetilde{h}$ is continuous on $[0,1]$, then $I_{\widetilde{h}}$ is p-quasi-concave if and only if $I_h$ is p-quasi-concave.
% 	\end{itemize}
 \end{lemma}
 \begin{proof}
 %	\underline{(i)} 
 	Suppose that $I_{\widetilde{h}}$ is p-quasi-convex.
 	%Denote by $a=h(0+)-h(0)>0$ and $b=h(1-)-h(1)>0$, and we have $I_{h}=aQ_1-bQ_0+I_{\widetilde{h}}$.
 	For $F,G\in\mathcal M$ and $\lambda\in(0,1)$, we have $Q_1(\lambda F+(1-\lambda) G))=\max\{Q_1(F), Q_1(G)\}$ and $ Q_0(\lambda F+(1-\lambda) G))=\min\{ Q_0(F), Q_0(G)\}$. Hence, by the p-quasi-convexity of $I_{\widetilde{h}}$, we have
 	\begin{align*}
 		I_h(\lambda F+(1-\lambda) G))&=a Q_0(\lambda F+(1-\lambda) G))-bQ_1(\lambda F+(1-\lambda) G))+I_{\widetilde{h}}(\lambda F+(1-\lambda) G))\\
 		&\le \min\{a Q_0(F), a Q_0(G)\}-\max\{b Q_1(F), bQ_1(G)\}+\max\{I_{\widetilde{h}}(F),I_{\widetilde{h}}(G)\}\\
 		&\le \max\left\{aQ_0(F)-bQ_1(F)+I_{\widetilde{h}}(F),aQ_0(G)-bQ_1(G)
 		+I_{\widetilde{h}}(G)\right\}\\
 		&=\max\{ I_{{h}}(F),I_{{h}}(G)\},
 	\end{align*}
 	which implies the p-quasi-convexity of $I_h$.
 \end{proof}

Using Lemma \ref{lm-dis01}, we can 
verify that scaled quantile-spread mixtures are p-quasi-convex. Hence,  both convex distortion functions and distortion functions in \eqref{eq:hQSM} lead to p-quasi-convexity of $I_h$. Combining with Proposition \ref{prop-transform}, the same conclusion also holds for $R_{h,v}$. We summarize these results in the following proposition.
Later we will show that they are the only possibilities.

\begin{proposition}\label{prop:trivial}
If $h\in \H^{\rm CX}\cup H^{\rm QSM}$ and $v:\R\to\R$, then $I_{h}$ and $R_{h,v}$ are both p-quasi-convex on $\mathcal M_c$.
\end{proposition} 
\begin{proof}
Note that $\{F\circ v^{-1}:F\in\mathcal M_c\}\subseteq \mathcal M_c$. It follows from Proposition \ref{prop-transform} that p-quasi-convexity of $I_h$ on $\mathcal M_c$ implies p-quasi-convexity of $R_{h,v}$ on $\mathcal M_c$. Therefore, we only need to consider the case of $I_h$.
By the definition of $I_h$ in \eqref{eq:1}, it is obvious that $I_h$ is p-quasi-convex if $h$ is convex.
%, as we have seen from Lemma \ref{th:www20} , convexity of $h$ implies p-convexity of $I_h$ which is stronger than p-quasi-convexity. 
Suppose that $h\in\mathcal H^{\rm QSM}$ which admits a representation as $I_h=S_{a,b}+k Q_{\alpha}^{c}$ with $a,b\ge 0$, $\alpha,c\in[0,1]$ and $k\in\R$. Since the mapping $k Q_{\alpha}^{c}$ is monotone and has convex level sets (e.g., \cite{WW20}), we know that $k Q_{\alpha}^{c}$ is both  p-quasi-convex and p-quasi-concave, which in turn implies that $I_h$ is p-quasi-convex by applying Lemma \ref{lm-dis01}.  Hence, we complete the proof.
\end{proof}

\subsection{The main result and a corollary on unbounded space}
The following theorem, which is the main technical result of the paper, establishes that the only possible generalized rank-dependent functions or signed Choquet functions with p-quasi-convexity are the ones with distortion functions in Proposition \ref{prop:trivial}, i.e., those with a convex distortion function or a distortion function in \eqref{eq:hQSM}. 
%When we mention Assumption \ref{assm:1}, the notation $x,y,z\in \R$ and the set of distributions $\mathcal M$ are defined in Assumption \ref{assm:1}.

\begin{theorem}\label{th-main}
Suppose that   Assumption \ref{assm:1} holds  and  $h\in\mathcal H^{\rm BV}$. The following statements are equivalent.
\begin{itemize}
\item[(i)] $h\in \mathcal H^{\rm CX}\cup\mathcal H^{\rm QSM}$.
\item[(ii)] $I_h$ is p-quasi-convex on $\mathcal M$.
\item[(iii)] $R_{h,v}$ is p-quasi-convex on $\mathcal M$ for some $v\in \mathcal V_{\M}$.
\item[(iv)] $R_{h,v}$ is p-quasi-convex on $\mathcal M$ for all functions $v$.
\end{itemize}
\end{theorem}

%\begin{remark}
%Necessity of Theorem \ref{th-main} is much more 
%challenging than sufficiency which is proved in Proposition \ref{prop:trivial}. We claim that necessity also holds if $I_h$ is  p-quasi-concave on the set of all simple distributions (the original space is $\mathcal M_c$ which contains all bounded distributions). This is because all distributions used in the proof of Theorem \ref{th-main} are simple distributions. More specifically, Lemma \ref{lm-monotone} and Lemma \ref{lm-main2}, which are two technical tools for the proof of Theorem \ref{th-main}, simply involve two point distributions and seven point distributions, respectively. 
%\end{remark}

%In Theorem \ref{th-main},we do not impose the continuity or the monotonicity on $v$, and we only require that $v$ can take at least three distinct values.
Because of Theorem \ref{th-main}, we will denote by $\mathcal H^{\rm QCX}=\mathcal H^{\rm CX}\cup\mathcal H^{\rm QSM}$,  which is the set of all $h\in \mathcal H^{\rm BV}$ for $R_{h,v}$ and $I_h$ to be p-quasi-convex.
Besides the class of convex distortion functions, the only other choices are  $h\in\mathcal H^{\rm QSM}$. 
One can immediately observe that all distortion functions in $\mathcal H^{\rm QSM}$ are neither continuous nor strictly monotone. 
%Therefore, we conclude that for a continuous or strictly increasing $h$, $I_h$ is p-quasi-convex if and only if $h$ is convex.  {\color{red}This yields to the equivalence  statement  for (vii) in  Theorem \ref{th-grand} in addition to the six conditions in Lemma \ref{th:www20}. (see )} 
 
%  illustrates that p-quasi-concavity is equivalent to the six conditions in Lemma \ref{th:www20}  if adding a natural assumption of continuity or strictly monotonicity of distortion functions.

The following proposition illustrates that 
%if $v$ can only take two distinct values,
if the domain of $R_{h,v}$ is chosen to be the set of all
two-point distributions on two specific points such that the values of $v$ are distinct on these two points,
%or $v$ can only take two distinct
then p-quasi-convexity of $R_{h,v}$ is equivalent to quasi-convexity of $h$. Hence, there are more cases of p-quasi-convex $R_{h,v}$ than in Theorem \ref{th-main}. 
%Hence, Assumption \ref{assm:1} on $\mathcal M$ is needed for a nontrivial characterization of p-quasi-convexity, and Theorem \ref{th-main} has characterized p-quasi-convexity of $R_{h,v}$ as general as can be.
 
\begin{proposition}\label{prop:two-point}
Let $x,y\in\R$ satisfying $x\neq y$, and define $\mathcal M=\{p\delta_x+(1-p)\delta_y: p\in[0,1]\}$. Suppose that $h\in \mathcal H^{\rm BV}$ and $v:\R\to\R$ satisfies $v(x)\neq v(y)$.
Then $R_{h,v}$ is  p-quasi-convex on $\mathcal M$ if and only if $h$ is quasi-convex.
\end{proposition} 
\begin{proof}
The	p-quasi-convexity of $R_{h,v}$ on $\mathcal M$ is equivalent to 
$$
R_{h,v}((\lambda p+(1-\lambda) q)\delta_x+(\lambda(1-p)+(1-\lambda)(1-q))\delta_y)\le \max \{R_{h,v}(p\delta_x+(1-p)\delta_y),R_{h,v}(q\delta_x+(1-q)\delta_y)\}
$$
for any $p,q,\lambda\in[0,1]$. Under some algebra calculations, this is equivalent to
$$
h(\lambda p+(1-\lambda)q)\le \max\{h(p),h(q)\}
$$
for any $p,q,\lambda\in[0,1]$.
Hence, we obtain the desired equivalence.
\end{proof}
Proposition \ref{prop:two-point} illustrates that Assumption \ref{assm:1} on $\mathcal M$ and the constraint $v\in\mathcal V_{\M}$ are needed for a nontrivial characterization of p-quasi-convexity.
Analogously to Proposition \ref{prop:two-point}, one can check that 
for  $h$  quasi-convex and $v$   only taking two values,  $R_{h,v}$ is p-quasi-convex on any convex set of distributions.

%Below we present a corresponding result to Theorem \ref{th-main} for the class of increasing distortion functions. Different from the strictly increasing case, there are nonconvex distortion functions that make $R_{h,v}$ p-quasi-convex.
%\begin{corollary}\label{co-increasing}
%Suppose that $\mathcal M$ satisfies Assumption \ref{assm:1}. For increasing $h\in\mathcal H^{\rm BV}$ and $v\in\mathcal V_{\M}$, $R_{h,v}$ is p-quasi-convex on $\mathcal M$ if and only if $h$ is convex or $h$ is a distortion function of $I_h=a Q_0+k Q_{\alpha}^c$ with some $a\ge 0$, $k\ge 0$, $\alpha\in[0,1]$ and $c\in[0,1]$.
%\end{corollary}

Below we present a corresponding result to Theorem \ref{th-main} that the generalized rank-dependent functions are defined on a general space that contains a distribution of unbounded random variables.  
%Denote by $\mathcal M_{\mathcal L}$ a convex set of distributions that satisfies Assumption \ref{assm:1} and contains a distribution of random variable, which is denoted by $X$, that is unbounded both from below and from above. Denote by 
%\begin{align*}
%	\mathcal H^{\mathcal L}=\{h\in\mathcal H^{\rm BV}: |R_{h,v}(F)|<\infty~{\rm for~all}~F\in\mathcal M_{\mathcal L}\}.
%\end{align*}
%which is the set of distortion functions suitable for $R_{h,v}$ to be real-valued on $\mathcal M_{\mathcal L}$.

%Denote by $\mathcal L$ a set of random variables containing all three-point random variables on three specific points, such that there exists $X\in\mathcal L$ such that $X$ is unbounded both from below and from above. We use $\mathcal M_{\mathcal L}$ to represent the set of distributions of all random variables in $\mathcal L$, and we assume that $\mathcal M_{\mathcal L}$ is convex. Denote by
%\begin{align*}
%\mathcal H^{\mathcal L}=\{h\in\mathcal H^{\rm BV}: |I_h(X)|<\infty~{\rm for~all}~X\in\mathcal L\},
%\end{align*}
%which is the set of distortion functions suitable for $I_h$ to be real-valued on $\mathcal L$.

\begin{corollary}\label{coro:gen}
Suppose that  Assumption \ref{assm:1} holds. Define a convex set $\mathcal M'$ such that  $\mathcal M\subseteq \mathcal M'$ and $\mathcal M'$ contains the distribution of a random variable, denoted by $X$, that is unbounded both from below and from above. Let $h\in\mathcal H^{\rm BV}$ and $v\in\mathcal V_{\M}$ satisfying $\essinf v(X)=-\infty$ and $\esssup v(X)=+\infty$. Then
$R_{h,v}:\mathcal M'\to\R$ is p-quasi-convex if and only if $h$ is convex and continuous or $h$ is the distortion function of mixed quantiles $I_h=kQ_{1-\alpha}^c$ for some $\alpha\in(0,1)$, $c\in[0,1]$ and $k\in\R$.
%Let $h\in\mathcal H^{\mathcal L}$ and $v\in\mathcal V_{\M}$. Suppose that $R_{h,v}$ is defined on $\mathcal M_{\mathcal L}$ which satisfies that $v(x),v(y),v(z)$ are distinct, and $\essinf v(X)=-\infty$ and $\esssup v(X)=+\infty$. Then
%$R_{h,v}$ is p-quasi-convex if and only if $h$ is convex and continuous or $h$ is the distortion function of mixed quantiles $I_h=kQ_{1-\alpha}^c$ for some $\alpha\in(0,1)$, $c\in[0,1]$ and $k\in\R$. 
\end{corollary} 

In the case of signed Choquet function ($v$ is the identity),
Corollary \ref{coro:gen} formalizes the observation that, since asymmetric spreads are not finite-valued on a general space with the distribution of unbounded random variables, we are left with only the class of convex distortion functions and those that correspond to scaled mixed quantiles.  

%As a consequence of Corollary \ref{coro:gen}, on $\mathcal M_{\mathcal L}$, 
%except for the case of a scaled mixed quantile, 
%  (i)-(vii) in  Section \ref{sec:2} are all equivalent.

%  A special case of $\mathcal L=L^q$, $q\in [1,\infty)$, i.e., the normed space of random variables with finite $q$-th moment, is summarized in the corollary below. 
% 
%\begin{corollary}
%Let $h\in\mathcal H^{  L^q}$, $q\in [1,\infty)$. 
%For a norm-continuous mapping $I_h:\M_{L^q}\to \R$, (i)-(vii) in  Theorem \ref{th-grand} are equivalent. 
%\label{coro:gen2}
%\end{corollary}
%Corollary \ref{coro:gen2} follows  from Corollary \ref{coro:gen} and the fact that mixed quantiles are not 
%norm-continuous.
%Different from  the other results, Corollary \ref{coro:gen2} imposes a continuity condition on $I_h$, instead of $h$,  to guarantee the equivalence of (i)-(vii).

\subsection{A unifying equivalence}
In this section, we present a unifying equivalence result on signed Choquet functions to illustrate the power of our main result.
In what follows,  the space is assumed to be $\M_c$ or $\X_c$, and $I_h$ is \emph{o-superadditive} if $I_h(X+Y)\ge I_h(X)+I_h(Y)$ for all $X,Y\in \X_c$. 

%\item Convexity: $I_h(\lambda X+(1-\lambda)Y)\le \lambda I_h(X)+(1-\lambda)I_h(Y)$ for all $X,Y\in L^\infty$ and $\lambda\in[0,1]$.
%\item Quasi-convexity: $I_h(\lambda X+(1-\lambda)Y)\le \max\{I_h(X),I_h(Y)\}$ for all $X,Y\in L^\infty$ and $\lambda\in[0,1]$.
%\item M-concavity: $I_h(\lambda F+(1-\lambda)G)\ge \lambda I_h(F)+(1-\lambda)I_h(G)$ for all $F,G\in \mathcal M_c$ and $\lambda\in[0,1]$.
%\item M-quasi-concavity: $I_h(\lambda F+(1-\lambda)G)\ge \min\{I_h(F),I_h(G)\}$ for all $F,G\in \mathcal M_c$ and $\lambda\in[0,1]$.
%\end{enumerate}

%From the perspective of mixing, convexity and quasi-convexity concer 

%Under some mild assumption on distortion functions, a surprising equivalence between the six conditions
%, which is the main result of this paper, 
%is shown in the following theorem.

\begin{theorem}\label{th-grand} 
If $h\in\H^{\rm BV}$ 
is continuous or strictly monotone, then the following are equivalent:
	(i) $h$ is convex;
	(ii) $I_h$ is concave-order monotone;
	(iii) $I_h$ is o-superadditive;
	(iv) $I_h$ is o-concave;
	(v) $I_h$ is o-quasi-concave;
	(vi) $I_h$ is p-convex;
	(vii)  $I_h$ is p-quasi-convex.
\end{theorem}

Theorem \ref{th-grand} can be broken down to a few pieces. The first six properties,  (i)-(vi) in Theorem \ref{th-grand}, are shown to be equivalent by \citet[Theorem 3]{WWW20} without any assumption on $h\in\H^{\rm BV}$.\footnote{\cite{WWW20} used the   narrative of risk measures, where results are stated for concave $h$; we seamlessly translate these results by a sign change.} %That is (i)-(vi) in Theorem \ref{th-grand} are equivalent for $h\in\mathcal H^{\rm BV}$. 
%We summarize this fact in the following lemma.
%% which will be used repeatedly in the paper. 
%
%\begin{lemma}[\cite{WWW20}]\label{th:www20} For $h\in\H^{\rm BV}$,  (i)-(vi) in Theorem \ref{th-grand} are equivalent.
%\end{lemma}
%Lemma \ref{th:www20}  
This result 
does not include p-quasi-convexity.
Indeed, p-quasi-convexity is not always equivalent to the above six conditions.  
%One may wonder whether p-quasi-concavity is also equivalent to the above six conditions.
In particular, any left or right quantile
is both p-quasi-convex and p-quasi-concave, but it is not o-concave or o-convex. %This explains the reason to state the assumption on distortion functions in Theorem \ref{th-grand}. 
Nevertheless, p-quasi-convexity and p-convexity are not too far away from each other.
\citet[Theorem 24]{W94} showed that (i) and (vii) in Theorem \ref{th-grand} are equivalent if $h\in\mathcal H^{\rm DT}$ is strictly increasing.\footnote{
	Under the same assumption,  \citet[Corollary 8]{WY21}
	showed the equivalence of (v) quasi-concavity  and (vii) p-quasi-convexity.} 
%\begin{lemma}[\cite{W94}]\label{th:w94} For strictly increasing $h\in\H^{\rm DT}$, 
%	(i) and
%	(vii) in Theorem \ref{th-grand} are equivalent. 
%\end{lemma}
%Lemma \ref{th:w94}  suggests that p-quasi-concavity is not far away from the other conditions, as soon as $h$ is strictly increasing.
%; under this condition the number of equivalent properties in consideration is increased from six to seven.
It is obvious that the normalization $h(1)=1$ does not matter, and thus $\H^{\rm DT}$ can be safely replaced by $\H^{\rm BV}$ in the above result. Note that $I_h=-I_{-h}$ for all $h\in\mathcal H^{\rm BV}$. Hence, for any strictly decreasing $h\in \H^{\rm BV}$, (i) and (vii) are also equivalent. Combining these arguments, we obtain the equivalence in Theorem \ref{th-grand} under the assumption that $h$ is strictly monotone. 
From a technical standpoint, strict monotonicity rules out any non-trivial considerations for signed Choquet functions over dual utilities (up to a sign change). 
 The remaining part of Theorem \ref{th-grand}, concerning continuous $h$, relies on Theorem \ref{th-main}.

Another result about the equivalence of (i)-(vii) is based on $L^q$, $q\in [1,\infty)$, i.e., the normed space of random variables with finite $q$-th moment. 
%Denote by
%\begin{align*}
%	\mathcal H^{L^q}=\{h\in\mathcal H^{\rm BV}: |I_h(X)|<\infty~{\rm for~all}~X\in L^q\}.
%\end{align*}

\begin{corollary}
	Let $h\in\mathcal H^{\rm BV}$, $q\in [1,\infty)$. 
	For a norm-continuous mapping $I_h:L^q\to \R$, (i)-(vii) in  Theorem \ref{th-grand} are equivalent. 
	\label{coro:gen2}
\end{corollary}
Corollary \ref{coro:gen2} follows  from Corollary \ref{coro:gen} and the fact that mixed quantiles are not norm-continuous.
Note that the same cannot be said on $L^\infty$ since all signed Choquet functions are norm-continuous on $L^\infty$ (e.g., \citet[Theorem 1]{WWW20}).
Different from  the other results, Corollary \ref{coro:gen2} imposes a continuity condition on $I_h$, instead of $h$,  to guarantee the equivalence of (i)-(vii).

\section{Proofs of Section \ref{sec:mqc} and further consequences}\label{sec:proof} %and Corollary  \ref{coro:gen}}
 
This section is dedicated to a proof of Theorem \ref{th-main}, which includes Theorem \ref{th-mainRDU} as a special case, and it further leads to the results in Corollaries %\ref{co-increasing}, 
\ref{coro:gen} and \ref{coro:gen2}. Moreover,   characterizations of p-quasi-concavity and p-quasi-linearity, which are  presented later,  also follow from Theorem \ref{th-main}.
%Several technical lemmas are needed, as the proof is much more involved than the case of strictly increasing ones. 

\subsection{Proofs of Theorem \ref{th-main} and its corollaries in Section \ref{sec:mqc}}
Note that (i) $\Rightarrow$ (ii), (iii), (iv) are proved in Proposition \ref{prop:trivial}, and it is obvious that (iv) $\Rightarrow$ (iii).
Next, we will show the assertion that (ii) $\Rightarrow$ (i) implies (iii) $\Rightarrow$ (i), and this means that it suffices to verify (ii) $\Rightarrow$ (i) for the completeness of the proof of Theorem \ref{th-main}. 
To see this assertion,
suppose that (ii) $\Rightarrow$ (i) holds.
Recall the set $\mathcal M$ that satisfies Assumption \ref{assm:1}.
We define the set of all $v$-transforms in $\mathcal M$ as 
\begin{align*}
	\mathcal M^v=\{F\circ v^{-1}: F\in\mathcal M\}\supseteq	\left\{p\delta_{v(x)}+(q-p)\delta_{v(y)}+(1-q)\delta_{v(z)}: 0\le p\le q\le1\right\}.
\end{align*}
It holds that $\mathcal M^v$ satisfies Assumption \ref{assm:1}.
By Proposition \ref{prop-transform}, if (iii) holds, then $I_h$ is p-quasi-convex on $\mathcal M^v$.
Using the result that (ii) $\Rightarrow$ (i) holds (note that the three distinct points in Assumption \ref{assm:1} can be chosen arbitrarily, and now let them be $v(x),v(y), v(z)$), we arrive at (i). Hence, (iii) $\Rightarrow$ (i) is verified.

%\renewcommand\theassumption{3P'}
%\begin{assumption}\label{assm:1}
%	The set $\M'$ is a convex subset of $\mathcal M_c$ and contains all three-point distributions on three specific points $x>y>z$, that is,
%	\begin{align}\label{eq-setassumption}
%		\left\{p\delta_x+(q-p)\delta_y+(1-q)\delta_z: 0\le p\le q\le1\right\}\subseteq \M' \subseteq\mathcal M_c.
%	\end{align}
%\end{assumption}

%\begin{theorem}\label{th-I_h}
%	Suppose that Assumption \ref{assm:1} holds and $h\in\mathcal H^{\rm BV}$. If $I_h$ is p-quasi-convex on $\mathcal M$, then $h\in\mathcal H^{\rm CX}\cup \mathcal H^{\rm QSM}$.
%\end{theorem}

%To see why Theorem \ref{th-I_h} can complete the proofs of (ii) $\Rightarrow$ (i) and (iii) $\Rightarrow$ (i),
%first, one can observe that Assumption \ref{assm:1} is stronger than Assumption \ref{assm:1}, and hence, (ii) $\Rightarrow$ (i) holds by applying Theorem \ref{th-I_h}. Second, recall the set $\mathcal M$ satisfies Assumption \ref{assm:1}.
%We define the set of all $v$-transforms in $\mathcal M$ as 
%\begin{align*}
%	\mathcal M^v=\{F\circ v^{-1}: F\in\mathcal M\}\supseteq	\left\{p\delta_{v(x)}+(q-p)\delta_{v(y)}+(1-q)\delta_{v(z)}: 0\le p\le q\le1\right\}.
%\end{align*}
%This means that $\mathcal M^v$ satisfies Assumption \ref{assm:1}. By Proposition \ref{prop-transform}, if (iii) holds, then $I_h$ is p-quasi-convex on $\mathcal M^v$ which further implies (i) by applying Theorem \ref{th-I_h}. 

In the following, we aim to prove (ii) $\Rightarrow$ (i) where several technical lemmas are needed. When we mention Assumption \ref{assm:1}, $x,y,z\in\R$ represent the distinct points defined in Assumption \ref{assm:1}, and we assume without loss of generality that $x>y>z$.
Suppose that $\mathcal M$ satisfies Assumption \ref{assm:1} and $h\in\mathcal H^{\rm BV}$.
We define a bivariate function as follows  
\begin{align}\label{eq-biquasifunction}
	\pi(p,q)= (x-y)h(p)+(y-z)h(q), ~~(p,q)\in T_2:=\{(a,b)\in [0,1]:a\le b\}.
\end{align} 
%where $x,y,z$ are defined in Assumption \ref{assm:1}, and assume without loss of generality that $x>y>z$. 
The first lemma shows that the p-quasi-convexity of $I_h$ on $\mathcal M$ implies that $\pi(p,q)$ is quasi-convex on $T_2$.

\begin{lemma}\label{prop:quasi-pi}
Suppose that $\mathcal M$ satisfies Assumption \ref{assm:1}. For $h\in\mathcal H^{\rm BV}$, if $I_h$ is p-quasi-convex on $\mathcal M$, 
then the function $\pi$ defined in \eqref{eq-biquasifunction} is quasi-convex on $T_2$.  In particular, $h$ is  quasi-convex on $[0,1]$.
\end{lemma}

\begin{proof}
	For any $0\le p\le q\le 1$, we have
	%Denote by $F_{p,q}=p\delta_x+(q-p)\delta_y+(1-q)\delta_z$ with $0\le p\le q\le 1$. We have 
	\begin{align*}
		I_h(p\delta_x+(q-p)\delta_y+(1-q)\delta_z)&=xh(p)+y(h(q)-h(p))+z(h(1)-h(q))\\
		&=zh(1)+(x-y)h(p)+(y-z)h(q)\\
		&=zh(1)+\pi(p,q).
	\end{align*}
	%To see necessity, suppose that $I_h$ is p-quasi-concave. We have for any $\lambda\in[0,1]$,
	%\begin{align*}
	%	I_h\left(\lambda F_{p_1,p_2}+(1-\lambda)F_{p_1',p_2'}\right)
	%	&=I_h\left(F_{\lambda p_1+(1-\lambda)p_1',\lambda p_2+(1-\lambda)p_2'}\right)\\
	%	&=zh(1)+\pi(\lambda p_1+(1-\lambda)p_1', \lambda (p_1+p_2)+(1-\lambda)(p_1'+p_2')).
	%\end{align*}
	One can verify that the p-quasi-convexity of $I_h$ implies  
	\begin{align}\label{eq-quasieq1}
		\pi(\lambda p_1+(1-\lambda)p_2,\lambda q_1+(1-\lambda)q_2)\le \max\{\pi(p_1,q_1), \pi(p_2,q_2)\}
	\end{align}
	for any $\lambda,p_1,p_2,q_1,q_2\in[0,1]$ with $p_1\le q_1$ and $p_2\le q_2$. 
	%Donote by $q=p_1+p_2$ and $q'=p_1'+p_2'$. \eqref{eq-quasieq1} is equivalent to \com{(12)? {\color{blue} Yes, I revise it.}}
	%\begin{align}\label{eq-quasieq2}
	%	\pi(\lambda p_1+(1-\lambda)p_1',\lambda q+(1-\lambda)q')\ge \min\{\pi(p_1,q), \pi(p_1',q')\}
	%\end{align}
	%for any $0\le p_1\le q\le 1$ and $0\le p_1'\le q'\le 1$.
	This is equivalent to the quasi-convexity of $\pi$ on $T_2$. Let $p=0$ in \eqref{eq-biquasifunction},   the quasi-convexity of $\pi$ implies the quasi-convexity of $h$ on $[0,1]$. This completes the proof.
\end{proof}

The second lemma is a generalization of Lemma 26 of \cite{W94}, who considered the strictly increasing distortion functions, to the set $\mathcal H^{\rm BV}$ and the domain $\mathcal M$ that satisfies Assumption \ref{assm:1}. For $h\in\mathcal H^{\rm BV}$ and $0\le p<q\le 1$ such that $h(p)\neq h(q)$,
we define $\lambda_h(p,q)$ as
\begin{align}\label{eq-lambda_h}
	\lambda_h(p,q)=\frac{h(p)/2+h(q)/2-h(p/2+q/2)}{|h(q)-h(p)|}.
\end{align}
The value of this function can be interpreted as a measure of the local convexity of $h$.

\begin{lemma}\label{lm-generalmain2}
Let $h\in\mathcal H^{\rm BV}$ and $0\le p<q\le s<t\le1$.
Suppose that $\mathcal M$ satisfies Assumption \ref{assm:1} and $I_h$ is p-quasi-convex on $\mathcal M$. If $h(p)\neq h(q)$, $h(s)\neq h(t)$ and $|h(q)-h(p)|(x-y)=|h(t)-h(s)|(y-z)$, then we have $\lambda_h(p,q)+\lambda_h(s,t)\ge 0$.
\end{lemma}

\begin{proof}
	%	By Lemma \ref{lm-monotone}, the distortion functions to be considered can be constrained into $\mathcal H^{\uparrow}\cup\mathcal H^{\downarrow}\cup\mathcal H^{\Lambda}$. For $h\in\mathcal H^{\uparrow}\cup\mathcal H^{\downarrow}\cup\mathcal H^{\Lambda}$, 

	Let $0\le p<q\le s<t\le1$ and $h(p)\neq h(q)$ and $h(s)\neq h(t)$. By Lemma \ref{prop:quasi-pi}, it holds that $h$ is quasi-convex on $[0,1]$. Thus, 
	one can check that only one of the following three cases will happen. \underline{Case 1}: $h(p)<h(q)\le h(s)<h(t)$; \underline{Case 2}: $h(p)>h(q)\ge h(s)>h(t)$; \underline{Case 3}: $h(p)>h(q)$ and $h(s)<h(t)$.
	
	Let us first verify Cases 1 and 2. Suppose that $0\le p<q\le s<t\le 1$,
	and $h(p)<h(q)\le h(s)<h(t)$ or $h(p)>h(q)\ge h(s)>h(t)$. All distributions in this proof are of the form 
	\begin{align}\label{eq-seven}
	p\delta_{x_1}+\frac{q-p}{2}\delta_{x_2}+\frac{q-p}{2}\delta_{x_3}+
	(s-q)\delta_{x_4}+\frac{t-s}{2}\delta_{x_5}+\frac{t-s}{2}\delta_{x_6}
	+(1-t)\delta_{x_7}
	\end{align}
	with fixed probabilities and $x_1,\dots,x_7\in\{x,y,z\}$ with $x_1\ge\dots \ge x_7$. Hence, we will simply use $F=(x_1,\cdots,x_7)$ to represent a distribution $F$, i.e., $F$ has the form in \eqref{eq-seven}.
	Define $F=(x,y,y,y,y,y,z)$ and $G=(x,x,x,y,z,z,z)$.  
%	  two distributions $F,G$  such that the values  $(x_1,\dots,x_7)$ of them equal to $(x,y,y,y,y,y,z)$ and $(x,x,x,y,z,z,z)$, respectively. 
One can verify that 
	$$
	I_h(F)-I_h(G)=(y-z)(h(t)-h(s))-(x-y)(h(q)-h(p)).
	$$
	Since $h(q)-h(p)$ and $h(t)-h(s)$ have the same sign, and $x>y>z$, the condition $|h(q)-h(p)|(x-y)=|h(t)-h(s)|(y-z)$ in the lemma implies $I_h(F)=I_h(G)$.
	By the p-quasi-convexity of $I_h$, we have $I_h((F+G)/2)\le(I_h(F)+I_h(G))/2$ where 
	$(F+G)/2=(x,x,y,y,y,z,z)$. This leads to
	$$
	\left(\frac{h(p)+h(q)}{2}-h\left(\frac{p+q}{2}\right)\right)(x-y)
	+\left(\frac{h(s)+h(t)}{2}-h\left(\frac{s+t}{2}\right)\right)(y-z)\ge 0.
	$$
	Dividing by the positive factors in  $|h(q)-h(p)|(x-y)=|h(t)-h(s)|(y-z)$, we have 
	$$
	\frac{1}{|h(q)-h(p)|}	\left(\frac{h(p)+h(q)}{2}-h\left(\frac{p+q}{2}\right)\right)
	+\frac{1}{|h(t)-h(s)|}\left(\frac{h(s)+h(t)}{2}-h\left(\frac{s+t}{2}\right)\right)\ge 0.
	$$
	This gives $\lambda_h(p,q)+\lambda_h(s,t)\ge 0$. To see Case 3,
	we only need to take $F=(x,y,y,y,z,z,z)$ and $G=(x,x,x,y,y,y,z)$,
%	replace the outcomes of constructed distributions $F,G$ by $(x,y,y,y,z,z,z)$ and $(x,x,x,y,y,y,z)$, respectively, 
and then, it follows from a similar proof of Cases 1 and 2 in the previous arguments.
\end{proof}

The third Lemma is divided into three cases to consider nonconvex distortion functions. We will apply Lemmas \ref{prop:quasi-pi}, \ref{lm-generalmain2} and some results in \cite{DK82} to prove this lemma.

\begin{lemma}\label{lm-generalupdown}
Let $h\in\mathcal H^{\rm BV}$.
Suppose that $\mathcal M$ satisfies Assumption \ref{assm:1} and $I_h$ is p-quasi-convex on $\mathcal M$. The following three statements hold.
	\begin{itemize}
		\item[(i)]	If $h$ is nonconvex on $(0,1)$, then there is a point $\alpha\in(0,1)$ such that $h$ is %discontinuous at $\alpha$ and takes
		a constant on both intervals $(0,\alpha)$ and $(\alpha,1)$, and the constant values are different.
		
		\item[(ii)] If $h$ is nonconvex on $[0,1)$  and convex on $(0,1)$, then $h(0+)>h(0)=0$, $h(p)=h(0+)$ for all $p\in(0,1)$ and $h(1-)\le h(1)$.
		
		\item[(iii)] If $h$ is nonconvex on $(0,1]$ and convex on $(0,1)$, then $h(1-)>h(1)$, $h(p)=h(1-)$ for all $p\in(0,1)$ and $h(0+)\le h(0)=0$.		
	\end{itemize}
\end{lemma}

\begin{proof}
	%	We consider the following three cases in turn to veirfy this result: \underline{Case 1}. $h$ is nonconcave on $(0,1)$; \underline{Case 2}. $h$ is nonconcave at $0$ and concave on $(0,1)$;
	%	\underline{Case 3}. $h$ is nonconcave at $1$ and concave on $(0,1)$.\\
	%\underline{Case 1}:
	%	For $p_1,p_2\in[0,1]$ and $p_1+p_2\le 1$, we have 
	%	\begin{align*}
		%		I_h(p_1\delta_x+p_2\delta_y+(1-p_1-p_2)\delta_z)&=xh(p_1)+y(h(p_1+p_2)-h(p_1))+z(h(1)-h(p_1+p_2))\\
		%		&=zh(1)+(x-y)h(p_1)+(y-z)h(p_1+p_2).
		%	\end{align*}
	%	Denote by $p=p_1$ and $q=p_1+p_2$, and define a bivariate function: 
	%	\begin{align}\label{eq-biquasifunction}
		%		\pi(p,q)= (x-y)h(p)+(y-z)h(q), ~~0\le p\le q\le 1.
		%	\end{align}
	By Lemma \ref{prop:quasi-pi}, the p-quasi-convexity of $I_h$ implies the quasi-convexity of $\pi$ defined in \eqref{eq-biquasifunction}. Hence, it is sufficient to prove this lemma based on the quasi-convexity of $\pi$.

\underline{(i)} We prove this statement in two steps (see Figure \ref{fig-mainlm1} for an illustration). First, we will show that there is at most one nonconvexity kink of $h$ on $(0,1)$. Second, we will show that $h$ should be a constant both on the left and the right of the nonconvexity kink.
\begin{figure}
{\centering
		\begin{tikzpicture}
			\draw[-] (0,5) -- (0,0) -- (8,0);
			\draw [black] (0,0) .. controls (1,0.05) and (3,0.4).. (4,2.5);
			\draw [black] (4,2.5) .. controls (5,2.6) and (6,3).. (7,4.5);
			\draw  [dash pattern={on 1.5pt off 2pt}] (3.5,1.7) -- (4.5,2.58);
\draw  [dash pattern={on 0.84pt off 2.51pt}] (3.45,1.2) -- (4,2.5)--(4.55,3.8)
node[above] {\footnotesize Slope is $K_L$};
\draw  [dash pattern={on 0.84pt off 2.51pt}] (2,2.2)--(4,2.5)--(6,2.8)
node[right] {\footnotesize Slope is $K_R$};			
			
\node at (4,2.5) { $\bigcdot$};	%alpha
\node at (3.3,2.85) {\footnotesize $(\alpha,h(\alpha))$};

\node at (2.5,0.75) { $\bullet$};  %t0
\node at (1.85,1.1) {\footnotesize $(t_0,h(t_0))$};

\node at (1.8,0.4) {$\bullet$};  %sn
\node at (1,0.65) {\footnotesize $(s_n,h(s_n))$};

\node at (3.5,1.7) {$\bullet$};  %pn
\node at (2.6,1.9) {\footnotesize $(p_n,h(p_n))$};

\node at (4.5,2.58) {$\bullet$};  %qn
\node at (5.5,2.25) {\footnotesize $(q_n,h(q_n))$};

\node at (4,2.12) {$\bullet$};  

\draw[very thick, dashed, <-](4.05,2.05)--(4.5,1.5)
node[right] {\footnotesize $\left(\alpha,\frac{h(p_n)+h(q_n)}{2}\right)$};
\end{tikzpicture} 
~\\[5pt]
}
\begin{center}
\footnotesize{ 
Step 1. There exists only one nonconvexity kink $\alpha$; $h$ is convex on $[0,\alpha)$ and on $(\alpha,1]$.
\\ Step 2. Contradiction from being nonconstant on $(0,\alpha)$. Here, $p_n=\alpha-1/n$,  $q_n=\alpha+1/n$, 
\\
$|h(t_0)-h(s_n)|(x-y)=|h(p_n)-h(q_n)|(y-z)$,  $K_{L}=h'_-(\alpha)$ and $K_{R}=h'_+(\alpha)$.
}
\end{center}
\caption{Graphic illustration of the proof of Lemma \ref{lm-generalupdown} (i).}\label{fig-mainlm1}
\end{figure}

Step 1: Suppose that $h$ is nonconvex on $(a,b)\subseteq(0,1)$,
	and we will show that there is at most one nonconvexity kink on $(a,b)$. To see this,  
	%for $p,q\in[0,1]$ and $p+q\le 1$, we have 
	%\begin{align*}
	%I_h(p\delta_x+q\delta_y+(1-p-q)\delta_z)&=xh(p)+y(h(p+q)-h(p))+z(h(1)-h(p+q))\\
	%&=zh(1)+(x-y)h(p)+(y-z)h(p+q).
	%\end{align*}
	%Denote by $p_1=p$ and $p_2=p+q$, and define a bivariate function: 
	%\begin{align}\label{eq-biquasifunction}
	%\pi(p_1,p_2)= (x-y)h(p_1)+(y-z)h(p_2), ~~0\le p_1\le p_2\le 1.
	%\end{align}
	%One can verify that the p-quasi-concavity of $I_h$ implies the quasi-concavity of $\pi$. 
	%Hence, 
	note that the function $\pi$ defined in \eqref{eq-biquasifunction} is quasi-convex both on $(0,a)\times(a,b)$ and  $(a,b)\times(b,1)$. Since $h$ is nonconvex on $(a,b)$, it follows from \citet[Theorem 2]{DK82}
	that $h$ is convex both on $(0,a)$ and $(b,1)$. Applying this result to smaller and smaller  subintervals of $(a,b)$, there is one point $\alpha\in(a,b)$ (the nonconvexity kink) such that $h$ is convex both on $(0,\alpha)$ and $(\alpha,1)$. 
	
Step 2: We aim to verify that $h$ is 
	%discontinuous at $\alpha$ and takes 
	a constant on both intervals $(0,\alpha)$ and $(\alpha,1)$, and this will complete the proof of (i). 
	%Assume now by contradiction that $h$ is nonconstant on intervals $(0,\alpha)$ or $(\alpha,1)$. 
	We only show that $h$ is a constant on $(0,\alpha)$ as the proof of the other case is similar. Assume now by contradiction that $h$ is nonconstant on $(c,d)\subseteq (0,\alpha)$. Let us consider the bivariate function $\pi$ (see \eqref{eq-biquasifunction}) on $(c,d)\times(d,1)$. Since $\alpha$ is a nonconvexity kink, we have that $h$ is nonconstant on $(d,1)$. Hence, $h$ is nonconstant both on $(c,d)$ and $(d,1)$. Note that $\pi$ is quasi-convex. It follows from
	\citet[Theorem 1]{DK82} that $h$ is continuous on $(d,1)$. Since $\alpha\in(d,1)$ and $h$ is convex on both $(0,\alpha)$ and $(\alpha,1)$, it holds that $h$ is continuous on $(0,1)$. 
	Now, recall $\lambda_h(p,q)$ defined in \eqref{eq-lambda_h}, and we 
	rewrite   it as
\begin{align}\label{eq-relambda_h}
		\lambda_h(p,q)=\frac{h(p)/2+h(q)/2-h(p/2+q/2)}{|h(q)-h(p)|}
		=\frac{f(p,q)-g(p,q)}{2|f(p,q)+g(p,q)|},
%	\frac{1}{4\left|{(h(q)-h(p))}/{(q-p)}\right|}\left(\frac{h(q)-h(p/2+q/2)}{(q-p)/2}
%		-\frac{h(p/2+q/2)-h(p)}{(q-p)/2}\right),~~p<q.
\end{align}
where
\begin{align*}
f(p,q)=\frac{h(q)-h(p/2+q/2)}{(q-p)/2}~~{\rm and}~~g(p,q)=\frac{h(p/2+q/2)-h(p)}{(q-p)/2}.
\end{align*}
By the convexity of $h$ on $(0,\alpha)$ and also noting that $h$ is nonconstant on $(0,\alpha)$, there exists $t_0$ in $(0,\alpha)$ such that $h$ has a nonzero left derivative at $t_0$. Therefore, it follows from \eqref{eq-relambda_h} that 
	\begin{align}\label{eq-leftdert_0}
		\lambda_h(s,t_0)\to0~~ {\rm as}~ s\uparrow t_0.
	\end{align}
	Denote  $K_{L}=h'_-(\alpha)$ and $K_{R}=h'_+(\alpha)$ by the left and right derivative of $h$ at $\alpha$. Since $\alpha$ is a point of nonconvex kink, we have $\infty\ge K_{L}>K_{R}\ge -\infty$. 	We assert that $K_LK_R\ge 0$, and this implies $|K_L+K_R|>0$. Otherwise, we have $K_L>0$ and $K_R<0$ which implies $h$ is not quasi-convex on $[0,1]$, and this contradicts to Lemma \ref{prop:quasi-pi}.
	Take $p_n=\alpha-1/n$ and $q_n=\alpha+1/n$ for $n\in\N$. 
Below we aim to verify the assertion that 
\begin{align}\label{eq-noncvkink}
\lim_{n\to\infty}\lambda_h(p_n,q_n)<0.
\end{align}
%	We show that $h(p_n)\neq h(q_n)$ for large enough $n$. To see this, if $K_L\le 0$, then $K_{R}<0$ which means that $h(p_n)> h(q_n)$ for large enough $n$. Suppose now $K_L>0$. By Lemma \ref{prop:quasi-pi}, $h$ is quasi-convex on $[0,1]$, and hence, we have $K_L>K_R\ge 0$. This implies $h(p_n)< h(q_n)$ for large enough $n$. Therefore, we conclude that $h(p_n)\neq h(q_n)$ for large enough $n$.
To see this,
by \eqref{eq-relambda_h}, we have
\begin{align*}
		\lambda_h(p_n,q_n)
		=\frac{a_n-b_n}{2|a_n+b_n|},
%		&=\frac{1}{4\left|{(h(q_n)-h(p_n))}/{(q_n-p_n)}\right|}\left(\frac{h(q_n)-h(p_n/2+q_n/2)}{(q_n-p_n)/2}
%		-\frac{h(p_n/2+q_n/2)-h(p_n)}{(q_n-p_n)/2}\right)\\
%		&=\frac{1}{4\left|{(h(q_n)-h(p_n))}/{(q_n-p_n)}\right|}\left[n\left(h\left(\alpha+1/n	\right)-h(\alpha)\right)-n\left(h(\alpha)-h\left(\alpha-1/n\right)\right)\right].
	\end{align*}
where
\begin{align*}
	a_n=n\left(h\left(\alpha+1/n	\right)-h(\alpha)\right)~~{\rm and}~~b_n=n\left(h(\alpha)-h\left(\alpha-1/n\right)\right).
\end{align*}
Note that $a_n\to K_R$ and $b_n\to K_L$ as $n$ tends to infinity. If  $\infty>K_L>K_R>-\infty$, then
\begin{align*}
\lambda_h(p_n,q_n)
=\frac{a_n-b_n}{2|a_n+b_n|}\to \frac{K_R-K_L}{2|K_R+K_L|}<0.
\end{align*}
%	Note that $(h(q_n)-h(p_n))/(q_n-p_n)$ is the slope of the line through $(p_n,h(p_n))$ and $(q_n,h(q_n))$. It holds that $(h(q_n)-h(p_n))/(q_n-p_n)\in[K_R,K_L]$ for large enough $n$ (see Figure \ref{fig-mainlm1}). Hence, when $n$ is large enough, we have
%	\begin{align}\label{eq-noncvkink}
%		\lambda_h(p_n,q_n)&\le \frac{1}{4\max\{|K_L|,|K_R|\}}\left[n\left(h\left(\alpha+1/n	\right)-h(\alpha)\right)-n\left(h(\alpha)-h\left(\alpha-1/n\right)\right)\right]\notag\\
%		&\to \frac{K_R-K_L}{4\max\{|K_L|,|K_R|\}}<0,~~{\rm as}~n\to\infty.
%	\end{align}
If $K_L=+\infty$ and $K_R\in\R$, then 
\begin{align*}
\lambda_h(p_n,q_n)
=\frac{a_n-b_n}{2|a_n+b_n|}=\frac{a_n/b_n-1}{2|a_n/b_n+1|}\to -\frac{1}{2}.
\end{align*}
If $K_L\in\R$ and $K_R=-\infty$, then 
\begin{align*}
	\lambda_h(p_n,q_n)
	=\frac{a_n-b_n}{2|a_n+b_n|}=\frac{1-b_n/a_n}{2|1+b_n/a_n|}\to -\frac{1}{2}.
\end{align*}
The case of $K_L=\infty$ and $K_R=-\infty$ cannot happen as $h$ should be quasi-convex on $[0,1]$ by Lemma \ref{prop:quasi-pi}.
Therefore, we have verified \eqref{eq-noncvkink}.
%We show that $h(p_n)\neq h(q_n)$ for large enough $n$. To see this, if $K_L\le 0$, then $K_{R}<0$ which means that $h(p_n)> h(q_n)$ for large enough $n$. Suppose now $K_L>0$. By Lemma \ref{prop:quasi-pi}, $h$ is quasi-convex on $[0,1]$, and hence, we have $K_L>K_R\ge 0$. This implies $h(p_n)< h(q_n)$ for large enough $n$. Therefore, we conclude that $h(p_n)\neq h(q_n)$ for large enough $n$.
Recall the definition of $t_0$ in \eqref{eq-leftdert_0}. Define a sequence $\{s_n\}_{n\in\N}$ such that $s_n< t_0$ for all $n\in\N$, $s_n\to t_0$ and $|h(t_0)-h(s_n)|(x-y)=|h(p_n)-h(q_n)|(y-z)$ for large enough $n$. Because $h$ is continuous on $(0,1)$, such a sequence  exists. Noting that $t_0<\alpha$, we have $s_n<t_0<p_n<q_n$ for large enough $n$.
	By Lemma \ref{lm-generalmain2}, we obtain $\lambda_h(s_n,t_0)+\lambda_h(p_n,q_n)\ge 0$ for large enough $n$.
	However, by combining  \eqref{eq-leftdert_0} and \eqref{eq-noncvkink}, we have $\lambda_h(s_n,t_0)+\lambda_h(p_n,q_n)<0$ for large enough $n$, and this yields a contradiction.
	Therefore, we conclude that $h$ should be a constant on both intervals $(0,\alpha)$ and $(\alpha,1)$, and this completes the proof of (i).
	
	\underline{(ii)}: Suppose that $h$ is nonconvex on $[0,1)$ and convex on $(0,1)$. This implies that $h(0+)>h(0)=0$.  It follows from Lemma \ref{prop:quasi-pi} that $h$ is quasi-convex on $[0,1]$. Combining this   with $h(0+)>h(0)$, it holds that $h$ is increasing on $[0,1]$, which implies $h(1-)\le h(1)$.
	It remains to verify that $h(p)=h(0+)$ for all $p\in(0,1)$. 
	%	By Proposition \ref{prop:quasi-pi}, we have $h$ is quasi-concave on $[0,1]$. Also noting that $h(0+)<h(0)$, we conclude that $h$ is decreasing on $[0,1)$.
	Assume now by contradiction that $h$ is nonconstant on $(0,1)$. Since $h$ is convex on $(0,1)$, there exists $\beta\in(0,1)$ such that $h$ is strictly increasing on $[\beta,1)$. Let $\epsilon>0$ and $k=2(x-y)/(y-z)$, and denote by $p_1=0$, $p_2=2\epsilon$, $q_1=\beta+2k\epsilon$ and $q_2=\beta$ such that $p_2\le q_2$ and $q_1<1$ (see the left panel of Figure \ref{fig-mainlm2}).
	 \begin{figure}
		\centering
		\begin{tikzpicture}
			\draw[-] (0,3.5) -- (0,0) -- (6,0);
			\draw [black] (0.02,0.5) .. controls (2,0.4) and (3,1.2).. (5,3.25);
			
			\node at (0,0) {$\bullet$};
			\node at (0,0.5) {$\circ$};
			
			\node[below] at (-0.2,0)  {\footnotesize $0$ $(p_1)$}; %p_1

			\node[below] at (3,0) {\footnotesize $\beta$ $(q_2)$};    %q_2
			\draw[gray,dotted] (3,0) -- (3,1.35);
			
			\node[below] at (1,0) {\footnotesize $2\epsilon$ $(p_2)$}; %p_2
			\draw[gray,dotted] (1,0) -- (1,0.55);
			
			\node[below] at (4.5,0) {\footnotesize $\beta+2k\epsilon$ $(q_1)$}; %q_1
			\draw[gray,dotted] (4.5,0) -- (4.5,2.78);
			
			\node at (0.5,0.46) {$\smallcdot$};
			\draw[very thick, dashed, <-](0.5,0.55)--(0.5,1.1)
			node at (0.75,1.45) {\footnotesize $\frac{p_1+p_2}{2}=\epsilon$};
			
			\node at (3.75,2) {$\smallcdot$};
			\draw[very thick, dashed, <-](3.72,2.1)--(3.4,2.5)
			node at (2.8,2.82) {\footnotesize $\frac{q_1+q_2}{2}=\beta+k\epsilon$};
		\end{tikzpicture} 
		~~~		
		\begin{tikzpicture}
			\draw[-] (0,-3.5) -- (0,0) -- (6,0);
			\draw [black] (0,0) .. controls (2,-2.45) and (3,-2.75).. (4.98,-2.65);
			\node at (5,-3.25) {$\bullet$};
			\node at (5,-2.65) {$\circ$};

			\node[above] at (5.4,0) {\footnotesize $1$ $(q_1)$};
			\draw[gray,dotted] (5,0) -- (5,-3.25);
			
			\node[above] at (4,0) {\footnotesize $1-2\epsilon$ $(q_2)$};
			\draw[gray,dotted] (3.8,0) -- (3.8,-2.63);
			
			\node[above] at (2.5,0) {\footnotesize $\beta$ $(p_2)$};
			\draw[gray,dotted] (2.5,0) -- (2.5,-2.25);
			
			\node[above] at (1,0) {\footnotesize $\beta-2k\epsilon$ $(p_1)$};
			\draw[gray,dotted] (1,0) -- (1,-1.1);

			\node at (4.4,-2.7) {$\smallcdot$};
			\draw[very thick, dashed, <-](4.4,-2.73)--(4.4,-3.25)
			node at (4.3,-3.6) {\footnotesize $\frac{q_1+q_2}{2}=1-\epsilon$};
			
			\node at (1.75,-1.83) {$\smallcdot$};
			\draw[very thick, dashed, <-](1.72,-1.85)--(1.35,-2.4)
			node at (1.37,-2.7) {\footnotesize $\frac{p_1+p_2}{2}=\beta-k\epsilon$};
			
		\end{tikzpicture} 
		\caption{Contradiction from $h$ being nonconstant on $(0,1)$ in the proof of (ii) and (iii) of Lemma \ref{lm-generalupdown}: the Left panel illustrates (ii) and the right panel illustrates (iii).}\label{fig-mainlm2}
	\end{figure}
We calculate the following items: 
	$$
	A_1:=\pi(p_1,q_1)=(y-z)h(\beta+2k\epsilon),
	$$
	$$
	A_2:=\pi(p_2,q_2)=(x-y)h(2\epsilon)+(y-z)h(\beta),
	$$ 
	and  
	$$
	A:=\pi\left(\frac{p_1+p_2}{2},\frac{q_1+q_2}{2}\right)=(x-y)h\left(\epsilon\right)+(y-z) h\left(\beta+k\epsilon\right).
	$$
	On the one hand,
	since $h(\epsilon)>h(0+)>0$ for any $\epsilon\ge 0$ and $h$ is continuous on $(0,1)$, we have $A>A_1$ for small enough $\epsilon>0$. On the other hand, 
	\begin{align*}
		A_2-A&=(x-y)\left(h(2\epsilon)-h\left(\epsilon\right)\right)
		-(y-z)(h(\beta+k\epsilon)-h(\beta))\\
		&=\epsilon(x-y)\left(\frac{h(2\epsilon)-h(\epsilon)}{\epsilon}-\frac{2(h(\beta+k\epsilon)-h(\beta))}{k\epsilon} \right)
		<0,
	\end{align*}
	where the inequality holds because $h$ is convex on $(0,1)$ and strictly decreasing on $[\beta,1)$. Therefore, we conclude that $A>\max\{A_1,A_2\}$ which contradicts the quasi-convexity of $\pi$. 
	This completes the proof of (ii).

	\underline{(iii)}: The proof of this statement is similar to (ii). Suppose that $h$ is nonconvex on $(0,1]$ and convex on $(0,1)$. This implies $h(1-)>h(1)$. By Lemma \ref{prop:quasi-pi}, we know that $h$ is quasi-convex on $[0,1]$. Combining with $h(1-)>h(1)$, we conclude that $h$ is decreasing on $[0,1]$. Hence, we have $h(0+)\le h(0)$. It remains to verify that $h$ is a constant on $(0,1)$.
	Assume now by contradiction that $h$ is nonconstant on $(0,1)$. Then, there exists $\beta\in(0,1)$ such that $h$ is strictly decreasing on $(0,\beta]$. Let $\epsilon>0$ and $k=2(y-z)/(x-y)$, and denote by $p_1=\beta-2k\epsilon$, $p_2=\beta$, $q_1=1$ and $q_2=1-2\epsilon$ such that $p_2\le q_2$ and $p_1>0$ (see the right panel Figure \ref{fig-mainlm2}). We calculate the following items: 
	$$
	A_1:=\pi(p_1,q_1)=(x-y)h(\beta-2k\epsilon)+(y-z)h(1),
	$$
	$$
	A_2:=\pi(p_2,q_2)=(x-y)h(\beta)+(y-z)h(1-2\epsilon),
	$$
	and 
	$$
	A:=\pi\left(\frac{p_1+p_2}{2},\frac{q_1+q_2}{2}\right)=(x-y)h\left(\beta-k\epsilon\right)+(y-z) h\left(1-\epsilon\right).
	$$
	On the one hand, since $h(1-\epsilon)\ge h(1-)>h(1)$ for any $\epsilon>0$ and $h$ is continuous on $(0,1)$, we have $A>A_1$ for small enough $\epsilon>0$. On the other hand, 
	\begin{align*}
		A_2-A&=(x-y)(h(\beta)-h(\beta-k\epsilon))-(y-z)(h(1-\epsilon)-h(1-2\epsilon))\\
		&=\epsilon(y-z)\left(\frac{2(h(\beta)-h(\beta-k\epsilon))}{k\epsilon}-\frac{h(1-\epsilon)-h(1-2\epsilon)}{\epsilon}\right)<0,
	\end{align*}
	where the inequality holds because $h$ is convex on $(0,1)$ and strictly increasing on $(0,\beta]$.
	Therefore, we conclude that $A>\max\{A_1,A_2\}$ which implies $\pi$ is not quasi-convex, a contradiction. Hence, we complete the proof of (iii). 
\end{proof}

In the following, we give the proof of (ii) $\Rightarrow$ (i) in Theorem \ref{th-main}, and this will complete the proof of Theorem \ref{th-main}.

\begin{proof}[Proof of (ii) $\Rightarrow$ (i) in Theorem \ref{th-main}]
%The implications (i) $\Rightarrow$ (ii) and (i) $\Rightarrow$ (iii) are shown in Proposition \ref{prop:trivial}. Let us now prove (ii) $\Rightarrow$ (i) and (iii) $\Rightarrow$ (i) below. 
%
%(ii) $\Rightarrow$ (i): 
It suffices to verify that for nonconvex $h\in\mathcal H^{\rm BV}$, the  p-quasi-convexity of $I_h$ implies $h\in\mathcal H^{\rm QSM}$. To see this, %it follows from Lemma %\ref{prop:quasi-pi} that $h$ is quasi-convex on $[0,1]$, and hence, if $h$ is convex on $[0,1)$ or $(0,1]$, then $h$ must be convex on $[0,1]$. Therefore,
we divide $h$ into three cases as shown in Lemma \ref{lm-generalupdown}: Case 1. $h$ is nonconvex on $(0,1)$; Case 2. $h$ is nonconvex on $[0,1)$ and convex on $(0,1)$; Case 3. $h$ is nonconvex on $(0,1]$ and convex on $(0,1)$.

By Lemma \ref{lm-generalupdown}, if $h$ is the form of Case 2 or Case 3, then one can check that $h\in\mathcal H^{\rm QSM}$. Thus, it remains to consider Case 1. 

Suppose that $h$ is nonconcave on $(0,1)$.
By Lemma \ref{lm-generalupdown} (i), there exists $\alpha\in(0,1)$ such that $h$ is a constant both on $(0,\alpha)$ and $(\alpha,1)$ and $h(\alpha-)\neq h(\alpha+)$. Hence, $h$ can be represented as
$$
h(p)=
h(0+)\id_{\{0<p<\alpha\}}+h(\alpha)\id_{\{p=\alpha\}}
+h(1-)\id_{\{\alpha<p<1\}}+h(1)\id_{\{p=1\}},~~p\in[0,1].
$$
In order to show that $h\in\mathcal H^{\rm QSM}$, we need to verify two assertions: (a) $h(0+)\le 0$ and $h(1-)\le h(1)$; (b) $(h(\alpha)-h(0+))(h(1-)-h(\alpha))\ge 0$. Both assertions will be proved by counter-evidence. For (a), we assume by contradiction that $h(0+)>0$ or $h(1-)>h(1)$. We only consider the case of $h(0+)>0$ as the case of $h(1-)>h(1)$ is similar.
By Lemma \ref{prop:quasi-pi}, we know that $h$ is quasi-convex, and combining with $h(0+)>0$, it holds that $h$ is increasing on $[0,1]$ which implies $h(0+)=h(\alpha-)<h(\alpha+)=h(1-)$. 
Let $p_1=0$, $p_2=\alpha/2$, $q_1=\alpha+2\epsilon$ and $q_2=\alpha-\epsilon$ such that $\epsilon>0$, $p_2\le q_2<q_1<1$ (see Figure \ref{fig-mainth1}).   \begin{figure}
	\centering
\begin{tikzpicture}
		\draw[-] (0,5) -- (0,0) -- (8,0);
		\draw [black] (0.02,2)--(4.47,2);
		\draw [black] (4.52,4)--(7,4);
		
\node at (0,0) {$\bullet$};
\node at (0,2) {$\circ$};
\node at (4.5,2) {$\circ$};
\node at (4.5,3) {$\bullet$};
\node at (4.5,4) {$\circ$};

\node at (-0.2,-0.3) {\footnotesize $0$ $(p_1)$};

\node at (2,-0.3) {\footnotesize $\frac{\alpha}{2}$ $(p_2)$};
\draw[gray,dotted] (2,0) -- (2,2);

\node at (3.5,-0.3) {\footnotesize $\alpha-\epsilon$ $(q_2)$};
\draw[gray,dotted] (3.5,0) -- (3.5,2);

\node at (4.5,-0.3) {\footnotesize $\alpha$};
\draw[gray,dotted] (4.5,0) -- (4.5,4);

\node at (6,-0.3) {\footnotesize $\alpha+2\epsilon$ $(q_1)$};
\draw[gray,dotted] (6,0) -- (6,4);

\node at (1,1.96) {$\smallcdot$};
\draw[very thick, dashed, <-](1,2.05)--(1,2.55)
node at (1,2.85) {\footnotesize $\left(\frac{\alpha}{4},h(0+)\right)$};

%\node at (1,2.3) {\footnotesize $\left(\frac{\alpha}{4},h(0+)\right)$};

\node at (4.75,3.96) {$\smallcdot$};
\draw[very thick, dashed, <-](4.75,4.05)--(4.75,4.6)
node at (5,4.9) {\footnotesize $\left(\alpha+\frac{\epsilon}{2},h(1-)\right)$};

	\end{tikzpicture} 
	\caption{Contradiction from $h(0+)>0$. }\label{fig-mainth1}
\end{figure}
Recall the function $\pi$ defined in \eqref{eq-biquasifunction}, we have  
\begin{align*}
	\pi(p_1,q_1)=(x-y)h(0)+(y-z)h(1-),~~\pi(p_2,q_2)=(x-z)h(0+),
\end{align*} 
and 
$$
\pi\left(\frac{p_1+p_2}{2},\frac{q_1+q_2}{2}\right)
=(x-y)h(0+)+(y-z)h(1-).
$$
One can check that
$
\pi\left((p_1+p_2)/{2},(q_1+q_2)/{2}\right)>\max\{\pi(p_1,q_1),\pi(p_2,q_2)\}
$
which implies $\pi$ is not quasi-convex, and this contradicts Lemma \ref{prop:quasi-pi}. Hence, we have $h(0+)\ge 0$. For (b), we assume by contradiction that $h(\alpha)<\min{\{h(0+),h(1-)\}}$ or $h(\alpha)>\max{\{h(0+),h(1-)\}}$. The case of $h(\alpha)>\max{\{h(0+),h(1-)\}}$ contradicts Lemma \ref{prop:quasi-pi} as $h$ should be quasi-convex. If $h(\alpha)<\min{\{h(0+),h(1-)\}}$,
let $p_1=\alpha/2$, $p_2=q_1=\alpha$ and $q_2=(1+\alpha)/2$ (see Figure \ref{fig-mainth2}).  
 \begin{figure}
	\centering
	\begin{tikzpicture}
		\draw (0,-3) -- (0,2);
		\draw (0,0) -- (8,0);
		\draw [black] (0.02,-1.5)--(3.97,-1.5);
		\draw [black] (4.02,1)--(7,1);
		
		\node at (0,0) {$\bullet$};
		\node at (0,-1.5) {$\circ$};
		\node at (4,-1.5) {$\circ$};
		\node at (4,-2.5) {$\bullet$};
		\node at (4,1) {$\circ$};
		
		\node at (-0.2,-0.3) {\footnotesize $0$};
		
		\node at (2,0.3) {\footnotesize $\frac{\alpha}{2}$ $(p_1)$};
\draw[gray,dotted] (2,0) -- (2,-1.5);
		
		\node at (4,-0.3) {\footnotesize $\alpha$ $(p_2, q_1)$};
\draw[gray,dotted] (4,-2.5) -- (4,1);
		
		\node at (5.5,-0.3) {\footnotesize $\frac{1+\alpha}{2}$ $(q_2)$};
\draw[gray,dotted] (5.5,0) -- (5.5,1);

		\node at (3,-1.53) {$\smallcdot$};
		\draw[very thick, dashed, <-](3,-1.55)--(3,-2.1)
		node at (3,-2.4) {\footnotesize $\left(\frac{3\alpha}{4},h(0+)\right)$};
		
	   \node at (4.75,1) {$\smallcdot$};
		\draw[very thick, dashed, <-](4.75,1.07)--(4.75,1.6)
		node at (4.75,1.9) {\footnotesize $\left(\frac{3\alpha}{4}+\frac{1}{4},h(1-)\right)$};
	\end{tikzpicture} 
	\caption{Contradiction from $h(\alpha)<\min\{h(0+),h(1-)\}$. }\label{fig-mainth2}
\end{figure}
 We have 
$$
\pi(p_1,q_1)=(x-y)h(0+)+(y-z)h(\alpha),~~\pi(p_2,q_2)=(x-y)h(\alpha)+(y-z)h(1-)
$$ 
and 
$$
\pi\left(\frac{p_1+p_2}{2},\frac{q_1+q_2}{2}\right)=(x-y)h(0+)+(y-z)h(1-).
$$ 
One can check that
$
\pi\left((p_1+p_2)/{2},(q_1+q_2)/{2}\right)>\max\{\pi(p_1,q_1),\pi(p_2,q_2)\}
$, which yields a contradiction to Lemma \ref{prop:quasi-pi}. Hence, we conclude that if $h$ is the form of Case 1, then $h\in\mathcal H^{\rm QSM}$. This completes the proof. 
\end{proof}

%\begin{proof}[Proof of Theorem \ref{th-I_h}]
%The implications (i) $\Rightarrow$ (ii) and (i) $\Rightarrow$ (iii) are shown in Proposition \ref{prop:trivial}. Let us now prove (ii) $\Rightarrow$ (i) and (iii) $\Rightarrow$ (i) below. 
%
%(ii) $\Rightarrow$ (i): It follows immediately from Theorem \ref{th-main} as Assumption \ref{assm:1} is stronger than Assumption \ref{assm:1}.
%
%
%(iii) $\Rightarrow$ (i): Suppose that (iii) holds.
%By Proposition \ref{prop-transform}, (iii) implies p-quasi-convexity of $I_h$ on $\mathcal M^v=\{F\circ v^{-1}: F\in\mathcal M\}$ which satisfies Assumption \ref{assm:1}. It then follows from Theorem \ref{th-main} that (i) holds. This completes the proof.
%\end{proof}

%\begin{proof}[Proof of Corollary \ref{co-increasing}]
%The result follows immediately from Theorem \ref{th-main} and the fact that increasing $h\in\mathcal H^{\rm QCX}$ are the ones shown in Corollary \ref{co-increasing}.
%\end{proof}

\begin{proof}[Proof of Corollary \ref{coro:gen}]
	The ``if" statement is straightforward to verify. 
	To show the ``only if" statement, we first note that, 
	since $\mathcal M'$ contains $\mathcal M$ which satisfies Assumption \ref{assm:1}, it follows from Theorem \ref{th-main} that $h\in\mathcal H^{\rm QCX}$. 
	Since the distribution of $X$ is in $\mathcal M'$, and $X$ is unbounded from below and above, we know that $h\in \H^{\rm BV}$ must be continuous at $0$ and $1$;  otherwise $Q_1(v(X))=\infty$ and $Q_0(v(X))=-\infty$ would lead to $R_{h,v}(X) \not \in \R$. This continuity condition for $h\in\mathcal H^{\rm QCX}$ leads to the   cases stated in the corollary.
\end{proof}

\begin{proof}[Proof of Corollary \ref{coro:gen2}] 
	Based on Corollary \ref{coro:gen}, it suffices to show that mixed quantiles are not 
	norm-continuous.  To verify this fact, take $\alpha \in(0,1)$ and $c\in [0,1]$, and a uniform random variable $U$ on $[0,1]$. 
	Let $X_n=a\id_{\{U\in [\alpha-\epsilon,\alpha+\epsilon]\}}+ b\id_{\{U> \alpha+\epsilon\}}  $ and $X=b\id_{\{U> \alpha+\epsilon\}}  $ 
	for some $0<a<b$ with $a\ne cb$.
	It is clear that $Q_\alpha^c(X ) =cb$, $Q_\alpha^c(X_n)  = a $ for each $n\in \N$, and $X_n\to X  $ in $L^q$. This example justifies the non-continuity of $Q_\alpha^c$.
\end{proof}

\subsection{P-quasi-concavity and p-quasi-linearity}

A considerable convenience to working with non-monotone $h\in \H^{\rm BV}$ is that we can put a negative sign in front of $h$ without leaving the class, which is  not the case for $h\in \H^{\rm DT}$. 
In particular, we have 
$R_{h,v}=-R_{-h,v}$ for $h\in\mathcal H^{\rm BV}$. Therefore,   $R_{h,v}$ is p-quasi-concave if and only if $R_{-h,v}$ is p-quasi-convex,
and all results on p-quasi-convexity immediately translate into results on p-quasi-concavity.  
% 
%  Before showing it, we denote by $\mathcal H^{\rm QCX}=\mathcal H^{\rm CX}\cup\mathcal H^{\rm -QSM}$
%where $\mathcal H^{\rm CX}=\{-h: h\in\mathcal H^{\rm CV}\}$ and $\mathcal H^{\rm -QSM}=\{-h: h\in\mathcal H^{\rm QSM}\}$. Specifically,
%%we need to introduce another two forms of $h$ in the subsets $\mathcal H^{\rm CX}$ and $\mathcal H^{\rm QCSM}$ of $\mathcal H^{\rm BV}$ that are important for the p-quasi-convexity property.
%\begin{itemize}
%\item[(i)] $h\in\mathcal H^{\rm CX}$: $h$ is convex.
%%\item[(ii)] $h\in\mathcal H_2^*$: For some $a,b\in \R$ such that either $a\ge 0$ or $b\ge0$, $h(p)=-a\id_{\{0<p\le 1\}}+b\id_{\{p=1\}}$, $p\in[0,1]$. In this case,
%%    \begin{align}
	%%    I_h=-a Q_{1}+b Q_{0}.
	%%    \end{align}
%\item[(ii)] $h\in\mathcal H^{\rm -QSM}$: For some $a\le 0$, $b\ge 0$ $\alpha\in[0,1]$, $c\in[0,1]$ and $k\in\R$, $h(p)=
%a\id_{\{0<p<\alpha\}}+(a+kc)\id_{\{p=\alpha\}}
%+(a+k)\id_{\{\alpha<p<1\}}+(a+b+k)\id_{\{p=1\}}$, $p\in[0,1]$. In this case,
%\begin{align}\label{eq-H2*}
%I_h=a Q_1+kQ_{1-\alpha}^{c}+b Q_0.
%\end{align}
%\end{itemize}
%
%A signed Choquet function with the form in \eqref{eq-H2*} is called a \emph{scaled quantile-converse spread mixture} as it is the sum of an asymmetric converse spread $aQ_1+bQ_0$ $(a\le 0,~b\ge 0)$ and a scaled mixed quantile
%$kQ_{1-\alpha}^{c}$ $(\alpha\in[0,1],~ c\in[0,1], ~k\in\R)$.
%%We denote by $\mathcal H^{\rm QCX}=\mathcal H^{\rm CX}\cup\mathcal H^{\rm QCSM}$.

\begin{corollary}\label{co-QCX}
	Suppose that  Assumption \ref{assm:1} holds. For $h\in\mathcal H^{\rm BV}$ and $v\in\mathcal V_{\M}$, $R_{h,v}$ is p-quasi-concave on $\mathcal M$ if and only if $h\in(-\mathcal H^{\rm QCX})$, that is, $h$ is concave or $-h\in\mathcal H^{\rm QSM}$.
	%either $h$ is concave, or $I_{-h}$ is a scaled quantile-spread mixture.
\end{corollary}

Note that p-quasi-linearity of a functional $\rho$ means that it is both p-quasi-convex and p-quasi-concave.
Combining Theorem \ref{th-main} and Corollary \ref{co-QCX}, a characterization of generalized rank-dependent functions with p-quasi-linearity is obtained.
\begin{corollary}\label{co-linearity}
	Suppose that  Assumption \ref{assm:1} holds. For $h\in\mathcal H^{\rm BV}$ and $v\in\mathcal V_{\M}$, $R_{h,v}$ is p-quasi-linear on $\mathcal M$ if and only if the signed Choquet function
	%$R_{h,v}(F)=I_h(F\circ v^{-1})$ where
	$I_h$ has one of the following forms.
	\begin{itemize}
		\item[(i)] $I_h=k\E$ for some $k\in\R$ where $\E$ represents the expectation.
		\item[(ii)] $I_h=k(cQ_{1}+(1-c)Q_{0})$ for some $k\in\R$ and $c\in[0,1]$.
		\item[(iii)] $I_h=kQ_{1-\alpha}^{c}$ for some $k\in\R$, $c\in[0,1]$ and $\alpha\in(0,1)$.
	\end{itemize}
\end{corollary}
\begin{proof}
	By Theorem \ref{th-main} and Corollary \ref{co-QCX}, we know that $R_{h,v}$ is p-quasi-linear if and only if $h\in\mathcal H^{\rm QCX}\cap\mathcal (-\mathcal H^{\rm QCX})$.
	It is straightforward to check that $R_{h,v}$  is one of the three forms in the corollary.
\end{proof}
Although p-quasi-linearity by definition is not related to the monotonicity of the distortion function,
all three forms of the distortion functions in Corollary \ref{co-linearity} are monotone. This is not surprising. Because p-quasi-linearity of $I_h$ implies quasi-linearity of $h$ (see Lemma \ref{prop:quasi-pi} and use a parallel result for the case of p-quasi-concavity), and a quasi-linear univariate function must be monotone.
\begin{remark}
	As a direct result of Corollary \ref{co-linearity}, only the three forms of $I_h$ in the corollary are possible to make $I_h$ p-quasi-linearity under Assumption \ref{assm:1}. This is not   new as \cite{WW20} showed a same result for $I_h$ with  CxLS (slightly weaker than p-quasi-linearity). Nevertheless, \cite{WW20} worked on a set containing all three-point distributions, a stronger condition than ours as the points are fixed in Assumption \ref{assm:1}.
\end{remark}

\section{A conflict between o-convexity and p-convexity}
\label{app:A}

In this section, we illustrate a conflict between o-convexity and p-convexity for constant-additive mappings; that is, a continuous and constant-additive mapping cannot be both o-convex and p-convex on $\X_c$ or $\M_c$ unless it is a multiple of the expectation. 
%Let $\X$ be the set of bounded random variables, and 
Recall that o-convexity of a mapping on $\M_c$ or $\X_c$ is defined as convexity on $\X_c$.
A mapping $\rho:\X_c\to \R$ is \emph{constant additive} if $\rho(X+c)=\rho(X)+\rho(c)$ for $X\in \X_c$ and $c\in \R$. 
On $\X_c$, continuity is with respect to the supremum-norm.
Continuous and constant-additive mappings on $\X_c$ include, but are not limited to, all signed Choquet functions and normalized monetary risk measures (\cite{FS16}).

\begin{proposition}\label{pr:conflict}
For a continuous and constant-additive mapping $\rho:\X_c \to \R$,
the following are equivalent:
 \begin{enumerate}[(i)]
 \item $\rho$ is o-convex and p-convex; 
 \item $\rho$ is o-concave and p-concave; 
 \item  $\rho=k\E $ for some $k\in\R$. 
 \end{enumerate}
 \end{proposition}
 \begin{proof}
Note that (i) and (ii) are symmetric, and (iii)$\Rightarrow$(i) is trivial. 
It suffices to show the direction (i)$\Rightarrow$(iii). 
 Since a p-convex mapping is necessarily law-based, we equivalently formula $\rho$ on $\M_c$. 
%  Let $c=\rho(0)=\widehat \rho(\delta_0)$, where
 Denote by $\delta_x$   the point-mass at $x\in \R$.  
Denote by $k=\rho(1)$. Note that    
  $\rho(x)=kx$ for $x\in \R$ since $\rho$ is constant additive and continuous.
Let $ F=\sum_{i=1}^n p_i \delta_{x_i}$ for some numbers $p_1,\dots,p_n\ge 0$ which add up to $1$ and $\{x_1,\dots,x_n\}\subseteq \R$. Let $X\sim F$. 
By p-convexity of $\rho$, we have 
$  \rho(F) \le \sum_{i=1}^n p_i   \rho(\delta_{x_i} ) = \sum_{i=1}^n p_i k  x_i   = k \E[X]  $. 
%Let $\tau:X\mapsto \rho(X)-\rho(0).$ 
 Since $\rho$ is  law-based, continuous and o-convex, it is convex-order monotone by e.g., the representation in \citet[Theorem 2.2]{LCLW20}.
Thus, $\rho(X)\ge \rho (\E[X]) = k\E[X]$ for all $X\in \X_c$. 
Putting the above two inequalities together, we have $\rho(X)=k\E[X] $ for all finitely supported random variables $X\in \X_c$.

For a general $X\in  \X_c$, 
let $X_n=   \lfloor  n X \rfloor /n$ for $n\in \N$, which is an approximation of $X$. It is clear that $X_n\to  X$ as $n\to \infty$, and $|X_n-X|\le 1/n$.
Using continuity of $\rho $ again,
 we have $\rho(X_n)\to \rho(X)$ as $n\to \infty$.
 Since $\rho(X_n) = k \E[X_n]   \to k \E[X]   $ as $n\to\infty$, we obtain   $\rho(X)=k \E[X]  $.  
 \end{proof}
 
 The same conclusion in Proposition \ref{pr:conflict} holds for  $\rho:L^p\to \R$ for $p\in [1,\infty)$ following the same proof.
 \begin{remark}Constant additivity of $\rho$ is essential for Proposition \ref{pr:conflict}.
A mapping $\rho: \X_c\to \R$ that is monotone, p-convex and o-convex does not need to be p-linear. 
For an example, take 
$
\rho_1:X\mapsto \E[f(X)]$ and $\rho_2:X\mapsto \E[g(X)]$ where $f$ and $g$ are two increasing  convex functions.
Clearly, $\rho_1$ and $\rho_2$ are both o-convex and p-linear. 
Since convexity is preserved under a maximum operation, the mapping $\rho:=\max \{\rho_1,\rho_2\}$ is o-convex and p-convex,
but it is not p-linear unless $f\ge g$ or $g\ge f$.
The reason that the proof does not work in this case is that, by letting $\ell(x)=\rho(x)$ for $x\in \R$, we can show using the argument above that 
$\ell(\E[X]) \le \rho(X)\le  \E[\ell(X)]$,
but this does not pin down $\rho$ unless $\ell$ is linear.
 \end{remark}

\section{Conclusion}\label{sec:7}

Probabilistic risk aversion (i.e., p-quasi-convexity) is characterized for rank-dependent utilities (Theorem \ref{th-mainRDU}) and generalized rank-dependent functions (Theorem \ref{th-main}). 
A new class of functionals, the mean-quantile mixtures, is shown to be the only class of  dual utilities that are p-quasi-convex and p-locally indifferent (Proposition \ref{th:mqm}).  
We have chosen to use p-convexity and o-concavity to present our main results, to be consistent with the literature on decision theory (e.g., \cite{Q93,W10}); by a simple sign change, we obtain corresponding results for p-concavity and o-convexity, a convention that is more common in the literature of risk management (e.g., \cite{MFE15, FS16}). 
Based on the characterization for generalized rank-dependent functions, we obtain a unified result of signed Choquet functions (Theorem \ref{th-grand}) containing  seven equivalent conditions for p-quasi-convexity. 
%Our results do not require strict increasing monotonicity. 
Our results are formulated for the more general objects, namely, signed Choquet functions and generalized rank-dependent functions. The corresponding results for dual utilities and rank-dependent utilities are also new, and our results help to understand classic decision models by disentangling monotonicity from other important properties.

\subsubsection*{Acknowledgments}  
The authors thank Mohammed Abdellaoui, David Budescu, Fabio Maccheroni,  Peter Wakker,  and Jingni Yang 
for their helpful discussions and bringing up relevant references.  
Ruodu Wang is supported by the Natural Sciences and Engineering Research Council of Canada (RGPIN-2024-03728) and Canada Research Chairs (CRC-2022-00141).

\end{document}